\documentclass[journal]{IEEEtran}
\usepackage{tikz}
\usetikzlibrary{patterns}
\usepackage{amsmath, amsfonts}  
\usepackage{etoolbox} 
\usepackage{enumitem}
\usepackage{cite}
\usepackage{algorithmic}
\usepackage{algorithm}
\usepackage{array}
\usepackage[caption=false,font=normalsize,labelfont=sf,textfont=sf]{subfig}
\usepackage{textcomp}
\usepackage{stfloats}
\usepackage{url}
\usepackage{verbatim}
\usepackage{graphicx}
\usepackage{authblk}

\usepackage{float}
\usepackage{xspace}
\usepackage[margin=1in]{geometry}
\usepackage{dsfont}
\usepackage{amssymb}
\usepackage{amsthm}
\usepackage{caption}
\usepackage[normalem]{ulem}
\usepackage[hidelinks]{hyperref}
\AtBeginDocument{%

}%
\usepackage{enumitem}
\setlist[enumerate]{itemsep=2pt,topsep=3pt,parsep=3pt,partopsep=2pt}

\newtheorem{lemma}{Lemma}

\newtheorem{defn}{Definition}

\newcommand{\1}[1]{\mathds{1}_{#1}}
\newcommand{\usher}{UShER\xspace}
\newcommand{\sarscov}{SARS-CoV-2}

\newcommand{\R}{\mathbb{R}}
\newcommand{\historydag}{\texttt{historydag}\xspace}

\DeclareMathOperator*{\argmin}{arg\,min}

\newcommand{\beginsupplement}{%
        \setcounter{table}{0}
        \renewcommand{\thetable}{S\arabic{table}}%
        \setcounter{figure}{0}
        \renewcommand{\thefigure}{S\arabic{figure}}%
     }

\title{The structure of deviations from maximum parsimony for densely-sampled data and applications for clade support estimation}

\author[1, 2, 3]{William Howard-Snyder}
\author[2]{Will Dumm}
\author[2]{Mary Barker}
\author[2]{Ognian Milanov}
\author[1]{Claris Winston}
\author[2]{David~H~Rich}
\author[4, 5, 6]{Marc A Suchard}
\author[2, 7, 8, 9]{Frederick A Matsen IV}
\affil[1]{Department of Computer Science and Engineering, University of Washington, Seattle, WA
}
\affil[2]{Computational Biology Program, Fred Hutchinson Cancer Research Center, Seattle, WA}
\affil[3]{Department of Computer Science, Princeton University, Princeton, NJ}
\affil[4]{Department of Human Genetics, University of California, Los Angeles}
\affil[5]{Department of Computational Medicine, University of California, Los Angeles}
\affil[6]{Department of Biostatistics, University of California, Los Angeles}
\affil[7]{Howard Hughes Medical Institute, Seattle, WA}
\affil[8]{Department of Genome Sciences, University of Washington, Seattle, WA}
\affil[9]{Department of Statistics, University of Washington, Seattle, WA}
\date{}

\begin{document}
\maketitle


\begin{abstract}
\emph{How} do phylogenetic reconstruction algorithms go astray when they return incorrect trees?
This simple question has not been answered in detail, even for maximum parsimony (MP), the simplest phylogenetic criterion.
Understanding MP has recently gained relevance in the regime of extremely dense sampling, where each virus sample commonly differs by zero or one mutation from another previously sampled virus.
Although recent research shows that evolutionary histories in this regime are close to being maximally parsimonious, the structure of their deviations from MP is not yet understood.
In this paper, we develop algorithms to understand how the correct tree deviates from being MP in the densely sampled case.
By applying these algorithms to simulations that realistically mimic the evolution of \sarscov{}, we find that simulated trees frequently only deviate from maximally parsimonious trees locally, through simple structures consisting of the same mutation appearing independently on sister branches.
We leverage this insight to design approaches for sampling near-MP trees and using them to efficiently estimate clade supports.
\end{abstract}


\section{Introduction}
\IEEEPARstart{A}{lthough} there have been hundreds of papers documenting the performance of phylogenetic algorithms, it can be difficult to understand \emph{why} phylogenetic algorithms do not succeed when they come up short.
Even for maximum parsimony (MP), the simplest phylogenetic criterion, the multitudes of equally parsimonious trees complicates assessment of error, because that phylogenetic error may depend substantially on the arbitrary choice of which MP tree to report.
Understanding MP has recently gained relevance in the regime of extremely dense sampling, in which each virus sample is commonly zero or one mutation away from another previously sampled virus.
Recent research has shown that the generating trees for data sets in this regime are close to being maximally parsimonious using theoretical~\cite{Wertheim2022-bm} and empirical~\cite{Kramer2023-sp} analysis.
However, the structure of these deviations from MP is not yet understood.

In this paper, we develop algorithms to understand the structure of deviations from maximum parsimony (MP) in simulations that realistically mimic the difficulty of phylogenetic reconstruction for \sarscov, and use this understanding to accurately estimate clade support from MP trees.
To do so, we introduce a novel summary statistic, \emph{parsimony diversity}, which enables us to generate data sets with realistic inferential difficulty.
We then develop algorithms to find the nearest MP tree among the many possible MP trees to the simulated tree, and compare the structure of this tree and its associated mutations to the simulated tree.
We find that differences between simulated trees and their most similar MP counterparts primarily consist of simple structures involving parallel child mutations on sister branches.
We explore one application of this observation to estimate clade support, a measure of confidence in particular groupings of observed taxa in the inferred tree, by sampling MP trees and randomly injecting subparsimonious structure.
Experiments on simulated data demonstrate that our support estimator yields more accurate support for simulated densely-sampled data than traditional approaches like the phylogenetic bootstrap or Bayesian clade support using existing implementations.

The steps of this procedure would not be possible without our recent development of a means of compactly representing highly parsimonious trees called the ``history sDAG'' (\autoref{fig:hsdagconstruction}).
Building on the matOptimize software~\cite{USHER2021}, we have developed Larch~\cite{larch}, an open-source software that can search for diverse MP trees on data and produce a history sDAG containing these trees.
Using the history sDAG and this software, we are often able to find well over $10^{10}$ highly parsimonious trees for SARS-CoV-2 data sets, even when we collapse branches without mutations and restrict to trees that are as parsimonious as any that the phylogenetic algorithm has found \cite{dumm2023representing}.
By enabling us to enumerate potentially all MP trees on simulated data, the history sDAG allows an efficient estimation of parsimony diversity, our summary statistic for comparing the difficulty of phylogenetic inference on simulated and real data.
The structure of the history sDAG containing all MP trees on a dataset facilitates dynamic programming algorithms that can rapidly find the closest MP tree to a given simulated tree, even when the number of MP trees is beyond astronomical.
The history sDAG also enables efficient sampling from this large set of MP trees.
In this paper, we show how uniform sampling from the trees in the history sDAG combined with a principled perturbation produces high accuracy clade support estimates.

\section{Methods}

\subsection{The history sDAG}
A key component of our analysis is a structure called the history sDAG.
We present an overview of the history sDAG here, but for further details, please refer to \cite{dumm2023representing}.
The history sDAG enables efficient storage of many phylogenetic histories, which are phylogenetic trees that include additional identifying data on internal nodes, such as inferred ancestral sequences.
In addition to its storage efficiency, this data structure has useful properties when storing MP histories.

When one constructs a history sDAG from a limited number of MP histories, the resulting structure often includes extra histories, and these are guaranteed to also be MP.
These new histories arise from swapping subhistories between the input histories, in a manner that preserves parsimony score.
This capability makes the history sDAG valuable for exploring and summarizing the landscape of MP histories in a dataset.
In this study, we use the history sDAG to store a collection of MP histories efficiently on simulated data and to swiftly identify the MP history most similar to the simulated history.
\begin{figure*}[!t]
\centering
\begin{tikzpicture}
    \def\firstcolx{5.6}
    \def\secondcolx{0}
    \def\thirdcolx{11}
    \def\toprowy{6.1}
    \def\labely{-3.5}

    \node[inner sep=0pt] (intreeone) at (\firstcolx, 0)
        {\includegraphics[width=0.28\textwidth]{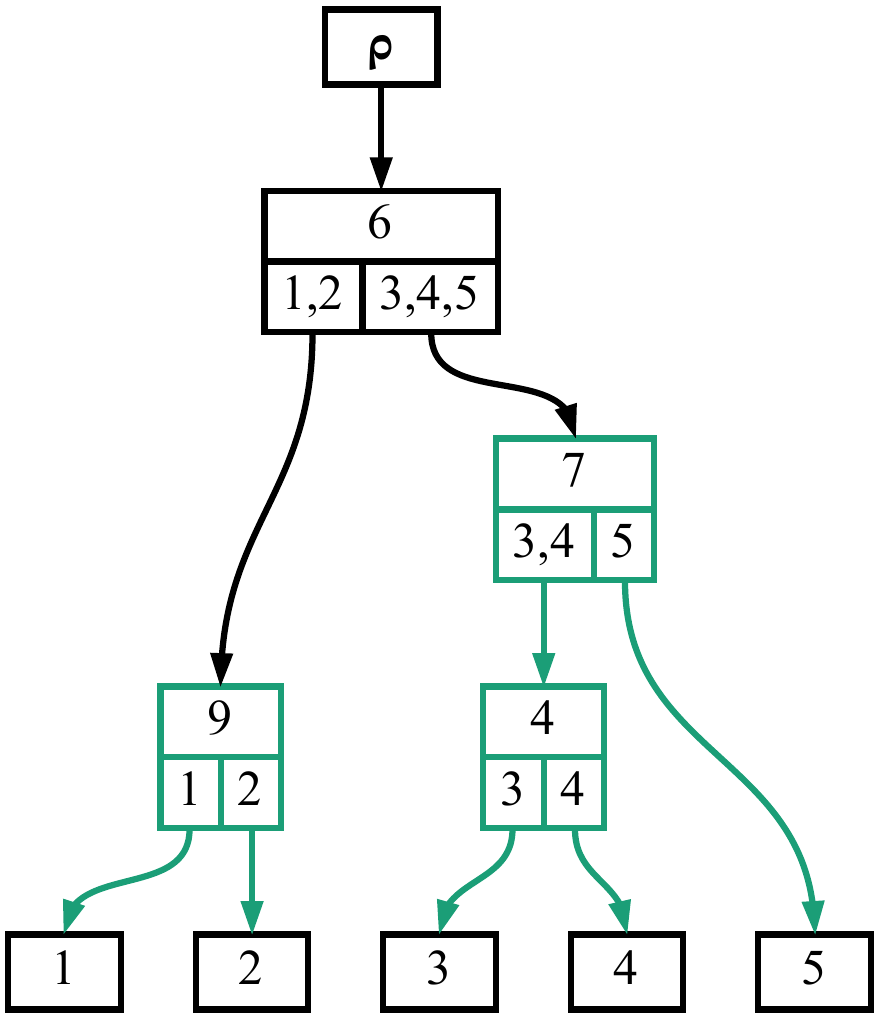}};
    \node[inner sep=0pt] (intreetwo) at (\firstcolx, \toprowy)
        {\includegraphics[width=0.28\textwidth]{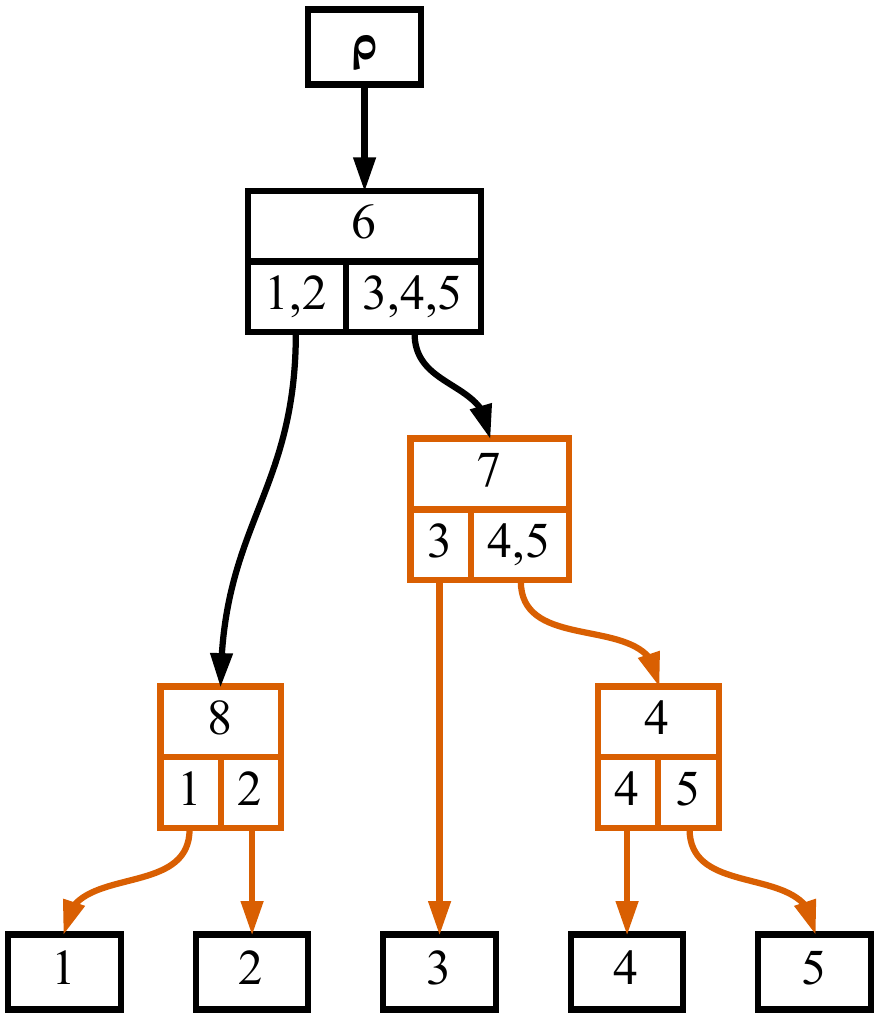}};
    \node[inner sep=0pt] (b) at (\firstcolx, \labely)
        {\small \textbf{(b)}};
    \node[inner sep=0pt] (dag) at (\secondcolx, 0)
        {\includegraphics[width=0.31\textwidth]{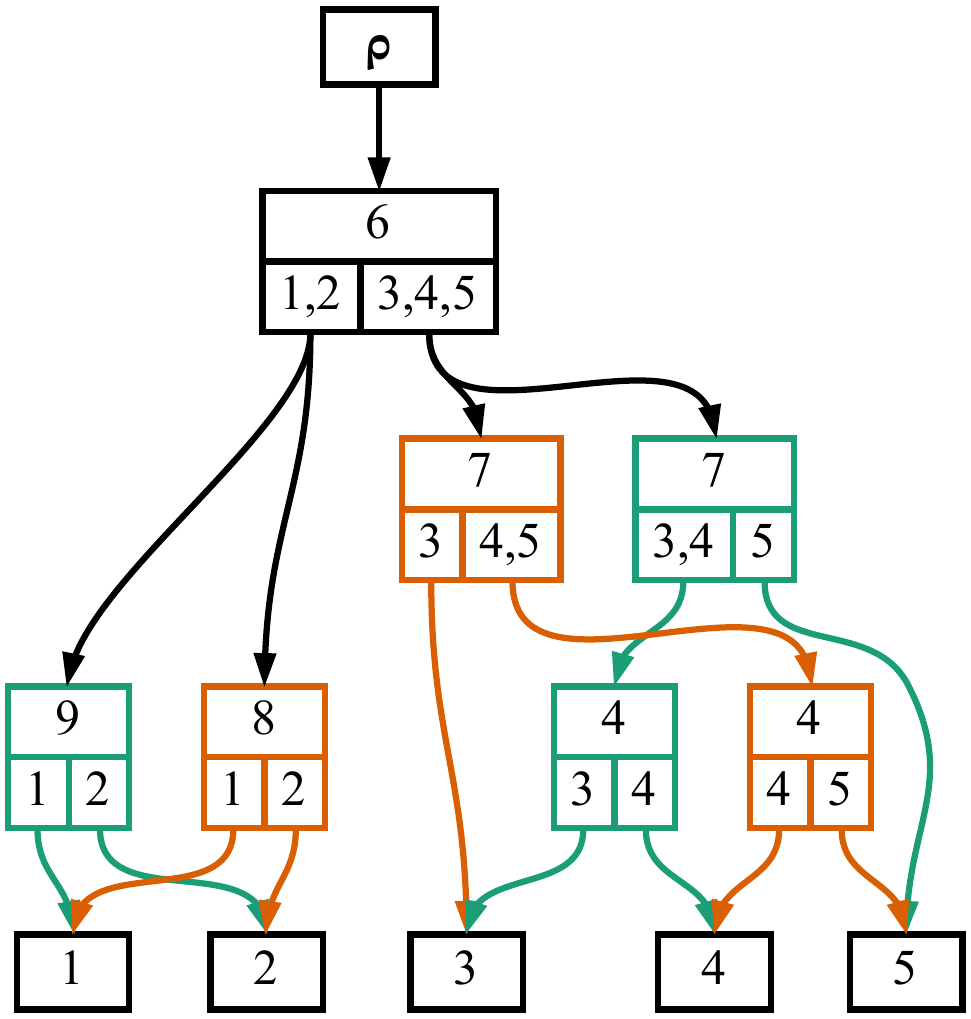}};
    \node[text width=3cm,align=center]  at (\secondcolx, \toprowy)
    {\small%
        $
        X \begin{cases}
            ~~1&~\text{CAC~~~~~~~}\\
            ~~2&~\text{AAG}\\
            ~~3&~\text{GCA}\\
            ~~4&~\text{GCT}\\
            ~~5&~\text{ACT}\\
        \end{cases}
        $
        \begin{align*}
            \,6&~~~\,\,\text{AAA}\\
            \,7&~~~\,\,\text{ACA}\\
            \,8&~~~\,\,\text{CAG}\\
            \,9&~~~\,\,\text{AAC}\\
        \end{align*}
        };
    \node[inner sep=0pt] (a) at (\secondcolx, \labely)
        {\small \textbf{(a)}};
    \node[inner sep=0pt] (extraone) at (\thirdcolx, 0)
        {\includegraphics[width=0.28\textwidth]{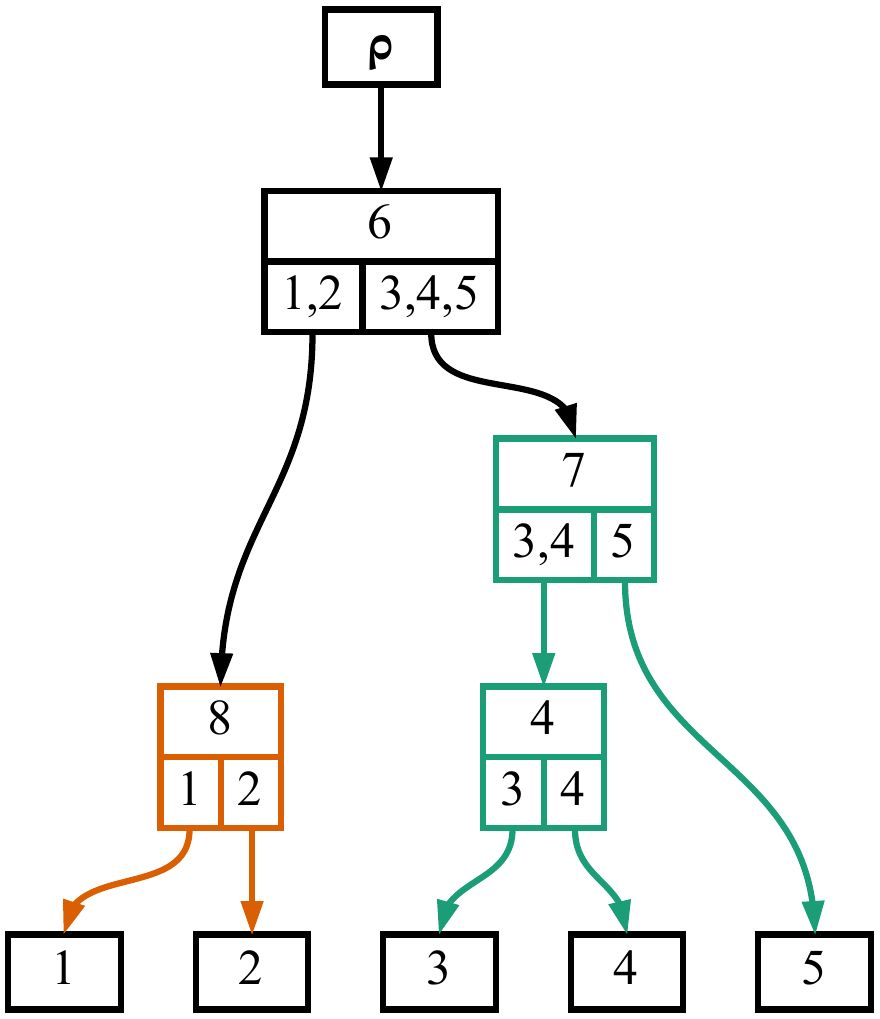}};
    \node[inner sep=0pt] (extratwo) at (\thirdcolx, \toprowy)
        {\includegraphics[width=0.28\textwidth]{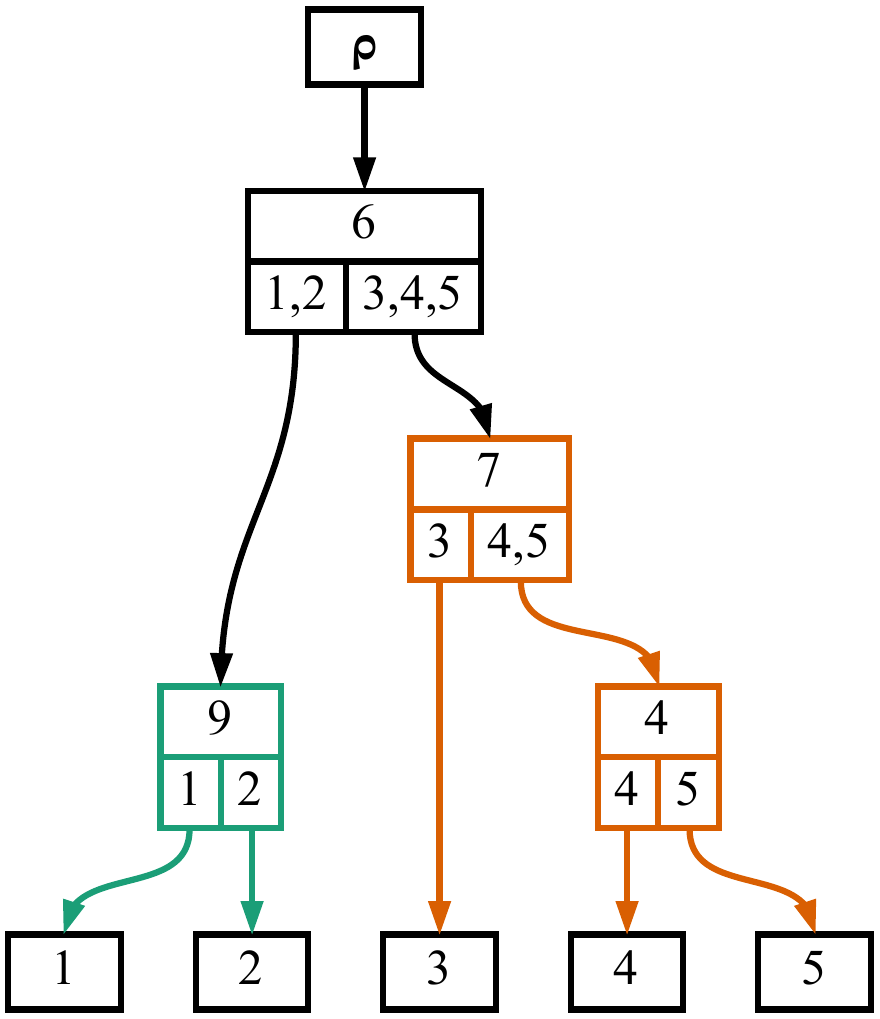}};
    \node[inner sep=0pt] (d) at (\thirdcolx, \labely)
        {\small \textbf{(c)}};
\end{tikzpicture}
\caption{\
    A sequence alignment $X $ and a history sDAG containing four histories relating the sequences of $X $ (a).
    Columns (b) and (c) show all histories contained in the history sDAG.
    A ``history'' is a tree structure such that each node consists of a label $\ell $ (shown on the top of each node), and a child clade set $U $, a set of sets of leaf node labels reachable from the node's children (shown on the bottom half of each node).
    The history sDAG is a graph union of the two histories in (b), allowing the green and orange subtrees to be swapped, creating the histories in (c).
    The top of each node shows the node label, which is a sequence, while the bottom shows child clade sets, with edges drawn to exit the child clade set which matches the target node's clade.
    Colors illustrate alternative subtrees, and are not part of the data structure.
}%
\label{fig:hsdagconstruction}
\end{figure*}

The history sDAG is simply a graph union of histories, augmented with redundant node information so that individual histories can be easily identified within the structure of the history sDAG.
Specifically, histories are defined so that the information augmented to each node consists not only of the observed and ancestral sequence data, but also the \emph{subpartition} of the node's clade formed by the clades of the node's children.
\autoref{fig:hsdagconstruction} shows an example of a simple history sDAG and the four histories it contains.

In preparation for Definition~\ref{defn:history}, which formalizes this information and the notion of histories in this paper, let $X $ be a set of aligned nucleotide sequences, and let $(V, E) $ be a directed graph with a \emph{universal ancestor (UA) node} denoted $\rho\in V$, and with the edge set $E \subset V\times V $.
For each node $v\in V \setminus \{\rho\} $, let $v $ be a pair $(\ell, U) $.
The sequence $\ell $ and the set $U $ are referred to as the \emph{label} and \emph{child clade set} of $v$, respectively.
If $v $ is a leaf node, $\ell\in X$ and $U = \emptyset $.
Let $C(v) \subset X$ be the \emph{child clade} of $v$, consisting of the set of labels of leaf nodes reachable from $v $.

\begin{defn}
    \label{defn:history}
    A directed acyclic graph $(V, E) $ with a universal ancestor node is a \emph{history} if it is tree-shaped, and if for any non-leaf node $v = (\ell, U) \in V $, $U$ is such that
    \begin{equation*}
        U = \left\{C(v_c) \mid v_c \in \text{Ch}(v) \right\},
    \end{equation*}
    where $\text{Ch}(v)$ are the children of $v$.

    A structure $(V', E') $ is a \emph{sub-history} of $(V, E) $ if $V'\subset V $ and $E'\subset E $, and these sets contain exactly one node $v\in V $ and all nodes and edges reachable from $v$ in $(V, E)$.
\end{defn}

Given this definition of a history, we can now concisely define the history sDAG.
\begin{defn}
    Let $T $ be a set of histories, and define
    \begin{align*}
        V' &= \bigcup_{(V, E) \in T} V \text{ and}\\
        E' &= \bigcup_{(V, E) \in T} E .
    \end{align*}
    Then $(V', E') $ is a \emph{history sDAG}.
    We may also say that $(V', E') $ is the history sDAG \emph{constructed from $T $}.
    For any $V_t \subset V' $ and $E_t \subset E' $ such that $(V_t, E_t) $ is a history, then we say that the history sDAG $(V', E') $ \emph{contains} the history $(V_t, E_t) $.
\end{defn}

Since the structure of a history determines the child clade set of each node, the definition of a history implicitly guarantees that edges have additional properties.
Notice that for any edge $e\in E $ in a history $(V, E) $, if $e = \left((\ell, U), v_c \right) $, then $C(v_c) \in U $.
Furthermore, for each node $v = (\ell, U) \in V $, and for each child clade $C \in U $, there must be exactly one edge $(v, v_c) \in E $ such that $C(v_c) = C $.
Additional edges descending from $v $ with this property cannot exist.
If they did, leaves with labels in $C $ would be reachable by multiple paths, which is impossible since a history is a tree.
This is an important property of edges in a history that allows us to recover individual histories from the history sDAG.

To sample a history from a history sDAG, one need only choose an edge descending from the UA node $\rho $, then choose one edge descending from each child clade of each node of the chosen edge's target node.
By continuing this process until the leaf nodes, which have no child clades, one obtains a history contained in the history sDAG.
In~\autoref{fig:hsdagconstruction}, edges are drawn exiting their parent node from the child clade that matches their target node's clade, and each history in~\autoref{fig:hsdagconstruction} can be seen as the result of this process of iteratively choosing edges starting from the UA node of the history sDAG.
We can sample uniformly from the collection of histories in the history sDAG by weighting our choice of edge descending from each clade proportional to the number of subtrees possible below the edge's target node.

\subsection{Minimum weight trimming}
\label{sec:min_weight_trim}

As a combinatorial object, the history sDAG can easily contain a vast number of histories, typically combinatorially larger than the set of histories from which it is constructed.
These additional histories result from swapping subhistories between the input histories, a simple example of which \autoref{fig:hsdagconstruction} shows.
In previous work \cite{dumm2023representing}, we showed how it is not uncommon for the history sDAG on real data to have 8 orders of magnitude more histories than the number of histories used to build it.

Conveniently, the structure of the history sDAG offers an efficient way of computing certain weights on those histories, and identifying the subset of histories which maximize or minimize such a weight.
We refer to this process of identifying optimal histories with respect to a weight as \emph{trimming} the history sDAG.
In \cite{dumm2023representing}, we describe a dynamic programming algorithm to trim the history sDAG so that it expresses only its minimum weight histories.
The algorithm involves removing all edges which point to suboptimal subhistories, and can be realized in two traversals of the history sDAG.
We define that algorithm in more detail below.

This efficient trimming algorithm can be applied for any history weight whose value decomposes as a sum over edges in each history.
More precisely,  let $T $ be the set of all histories contained in history sDAG $(V, E)$, and let $g: T \to \R $ be a function which assigns weights to histories.
For any history $t\in E $, let the set of edges in $t $ be denoted $E_t $.
The trimming algorithm can be used if there exists an edge weight function $f:E\to \R$ so that $g(t) = \sum_{e\in E_t} f(e) $, for any history $t\in T $.

Given such a history weight, we can compute the minimum weight of any subhistory beneath a node $v\in V $, which we will denote $M_f(v) $, in a single post-order traversal of the history sDAG \cite{dumm2023representing}.
Denote this operation $\textbf{AnnotateMinWeight}(V, E)$.
The full trimming algorithm is presented in the Supplementary Materials (\autoref{sec:min_trim_pseudocode}), and involves removing edges from the history sDAG which do not minimize subhistory weight below a node.
The output of this algorithm, which we will call $\textbf{MinTrim}(V, E, f) $, is a new history sDAG that contains only those histories in the history sDAG $(V, E) $ that minimize the weight function $g $.

This trimming algorithm is essential for our methods in two ways.
First, it allows us to trim the history sDAG to contain only histories that minimize the parsimony score, which is a sum of mutations on each edge in a history.
We also use this algorithm to trim a history sDAG containing MP histories to contain only histories minimizing Robinson-Foulds (RF) distance to a fixed reference history.
In the Supplementary Material (~\autoref{sec:rf_distance_decomposes}) we describe how RF distance can be decomposed as a sum over edges, making this application of the algorithm possible.

\subsection{MP-history search with the history sDAG}
\label{sec:hdag_mp_tree_search}

Many of our algorithms depend on producing and efficiently storing a large collection of MP histories.
We describe this subroutine here.
Given a set of aligned sequences $X$ and a number of iterations $n$, we execute the following steps:
\begin{enumerate}
  \item Build a tree from the sequences in $X$ (with \usher{}) by greedily placing leaves one at a time so as to minimize the parsimony score.
  \item Build a history sDAG $(V, E)$ from that tree.
  \item Repeat the following steps $n$ times.
  \begin{enumerate}
    \item Sample a history $t$ uniformly from $(V, E)$.
    \item Greedily select SPR moves (with Larch) to minimize the parsimony score of $t$.
    \item Graph union the result to $(V, E)$.
  \end{enumerate}
  \item Return $\textbf{MinTrim}(V, E, f)$ where $f: E \rightarrow \mathbb{Z}$ counts the number of mutations on edge $e$.
\end{enumerate}
Sampling a history in step 3a and adding it to the history sDAG in step 3c take time proportional to the number of nodes, which is $O(X)$.
Optimizing the history greedily in step 3b takes $O\left(X^2 \log X\right)$ in the worst case, but Larch caches Fitch sets and uses parallelism to make it run very fast in practice.
The overall runtime of this algorithm is therefore $O\left(E + n X^2 \log X\right)$ and the space complexity is $O(E + V)$.
In our experiments, we use $n = 20,000$.
We refer to the trees contained in the trimmed history sDAG as \emph{inferred-MP histories}.
Recall that although we may only add up to $n$ histories to the history sDAG, the number of histories discovered is often much larger than $n$ due to subhistory swapping.
This allows our algorithm to find many orders of magnitude of inferred-MP histories, while maintaining space and memory efficiency.
For instance, this procedure found over $10^{20}$ inferred-MP histories on the \usher{}-clade A.2.5 with 620 unique sequences.

Steps 2 and 3 are implemented in Larch, an open-source C++ implementation of the history sDAG at \url{https://github.com/matsengrp/larch}.
We implemented step 4 using the Python implementation of the history sDAG at \url{https://matsengrp.github.io/historydag}.
All bash scripts, Python experiments, code for generating plots, and instructions on how to reproduce results can be found in our GitHub repository at \url{https://github.com/matsengrp/hdag-benchmark}.

The MP trees we attempt to find, and those stored in the resulting history sDAG are \emph{collapsed}, in the sense that all non-pendant edges must have mutations on them.
We make this choice to avoid storing nodes which result from arbitrary binary resolutions of clades, when there is no signal in the sequence data to support a particular resolution.
Therefore, the true number of binary MP trees is much larger than the number of MP trees stored in the history sDAG.
Supposing that all collapsed MP histories were contained in the history sDAG however, any binary MP history could be realized as a resolution of a collapsed history, with resolved edges containing no mutations.

\subsection{Parsimony diversity}
\label{sec:simulations}

Real sequence data often support many phylogenetic trees, making it challenging to accurately infer the true generating tree.
Given that we want our claims about the simulated data to generalize to real data, we  ensure that phylogenetic inference on the simulated data is sufficiently difficult.
We use Parsimony Diversity (PD) as a proxy for how hard an inference problem is.
Parsimony Diversity is defined as \emph{the number of unique clades among fully-collapsed MP trees}.
Under the MP criterion, these are the groupings of taxa that the data supports.
Additionally, in the densely sampled regime, the MP criterion and ML criterion converge \cite{Wertheim2022-bm}, so the clades that appear in MP trees should be a reasonable summary of the most likely clades in our setting.
Previous work shows that similar summary statistics measuring the diversity of MP trees are the most important predictors of the difficulty of phylogenetic inference \cite{haag2022}.
In the sections below, we describe how we use Parsimony Diversity to tune the inferential difficulty of our simulated datasets.

\subsubsection*{\bf Simulations}
Following previous work \cite{MAPLE2022}, we start with the background tree on over 2 million \sarscov{} genomes at \url{http://hgdownload.soe.ucsc.edu/goldenPath/wuhCor1/UShER_SARS-CoV-2/2022/10/01} \cite{mcbroome2021} and inferred using \usher{} \cite{USHER2021}.
We first select subtrees containing between 500 and 1000 (non-unique) leaves.
We must also fully resolve these subtrees because directly simulating on the \usher{} subtrees would create too many multifurcations from edges without mutations.
Then, for each subtree we:
\begin{enumerate}
    \item Randomly resolve all multifurcations and give new edges small branch lengths of 1/1000th the length of a branch with a single mutation.
    \item Use phastSim \cite{2022phastsim} to simulate sequences on the tree starting with the Wuhan-1 reference sequence.
          We repeat this step until every leaf sequence forms a monophyletic clade with its duplicate sequences.
    \item Retain a subset of the leaves so that leaves are uniquely labeled by their sequences.
    \item Write an alignment containing all the unique leaf sequences in this tree.
\end{enumerate}

In step 2, we simulate using the UNREST substitution model with empirically determined rates, Gamma distributed rate variation, and hypermutation.
On average, the total tree length was 0.032 and the substitution rate per site was 1.
We include hypermutation to better fit the patterns of mutability observed in \sarscov{}, which indicate that some sites have much larger base-to-base mutation rate than others, possibly reflecting selective pressure \cite{demaio2021}.
PhastSim \cite{2022phastsim} models this effect by setting sites as hypermutable with probability $p$ and scaling a randomly selected base-to-base mutation rate $r$.
We set $p=0.01$ and use a grid search to tune the hypermutation rate parameter $r$ and the shape parameter $\alpha$ of the Gamma distribution to reflect the Parsimony Diversity observed in real data.

The goal is to match the PD of the simulated data to that of the real data.
However, PD tends to increase with the number of unique sequences in the dataset.
Therefore, we compare the PD of the data simulated on an \usher{} subtree to the dataset consisting of the observed sequences in that subtree.
Still, the simulated dataset could differ in size from the corresponding real one, since we only use unique sequences.
To correct for this, we also divide the number of nodes inferred on each dataset by the size of the dataset before comparing the two.
We call this metric Corrected Increase in Parsimony Diversity (CIPD) and define it as
\begin{equation}
\label{eq:corrected_pd_increase} 
    \text{CIPD} = \frac{\text{Simulation PD}}{\text{Real data PD}} \cdot \frac{\text{Number of tips in real data}}{\text{Number of tips in simulation}}.
\end{equation}
Compared to real data, simulated data have more MP groupings of taxa per taxon when CIPD is larger than 1, and fewer when CIPD is less than 1.
When CIPD is exactly 1, that means the simulated data and real data admit similarly diverse sets of MP trees.

We use CIPD as a proxy for the relative increase in difficulty between real and simulated data.
Previous work that simulates data in the densely sampled regime \cite{MAPLE2022} does not evaluate the inferential difficulty of their simulated data, and our analysis in \autoref{sec:cipd_maple} shows that the PD of their simulated data was on average lower than the PD of real data.
In our work, we use a nearly identical simulation procedure, but tune the parameters in our simulation to make the inference problem sufficiently difficult.

\subsubsection*{\bf Tuning simulation hyperparameters}
In total, for each of the 16 parameter combinations, we extracted 8 different \usher{} subtrees.
For each subtree, we conducted 5 trials, where each trial consisted of randomly resolving the subtree, simulating sequences along its branches with phastSim \cite{2022phastsim}, and computing the CIPD for that trial.
We compute the CIPD by producing a large collection of MP trees stored in the history sDAG, as described in \autoref{sec:hdag_mp_tree_search}, then finding all the unique clades among histories represented in the history sDAG.
This can be done efficiently by looking at each node $v = (\ell, U)$ and storing the union of child clades $\cup_{C \in U} C$ in a set.
Note that it is unlikely that our history sDAG contains all MP trees for the simulated or real datasets, thus we are underestimating the PD for both.
However, since we have no reason to think that this error will be larger for real data than simulated (or vice versa), we expect this procedure generates a reasonable estimate for CIPD.
We find that with $\alpha = 0.1$ and $r = 50$ the simulation produces a realistic CIPD of 0.97.
The median CIPD value for each combination of $\alpha$ and $r$ are shown in \autoref{fig:pd_heatmap}.

\section{Results}

\subsection{Identifying subparsimonious structures in trees}
\label{sec:identifying_pcms}

\begin{figure*}[!t]
    \centering
    \includegraphics[width=.75\textwidth, trim={0cm 0cm 0cm 0cm}]{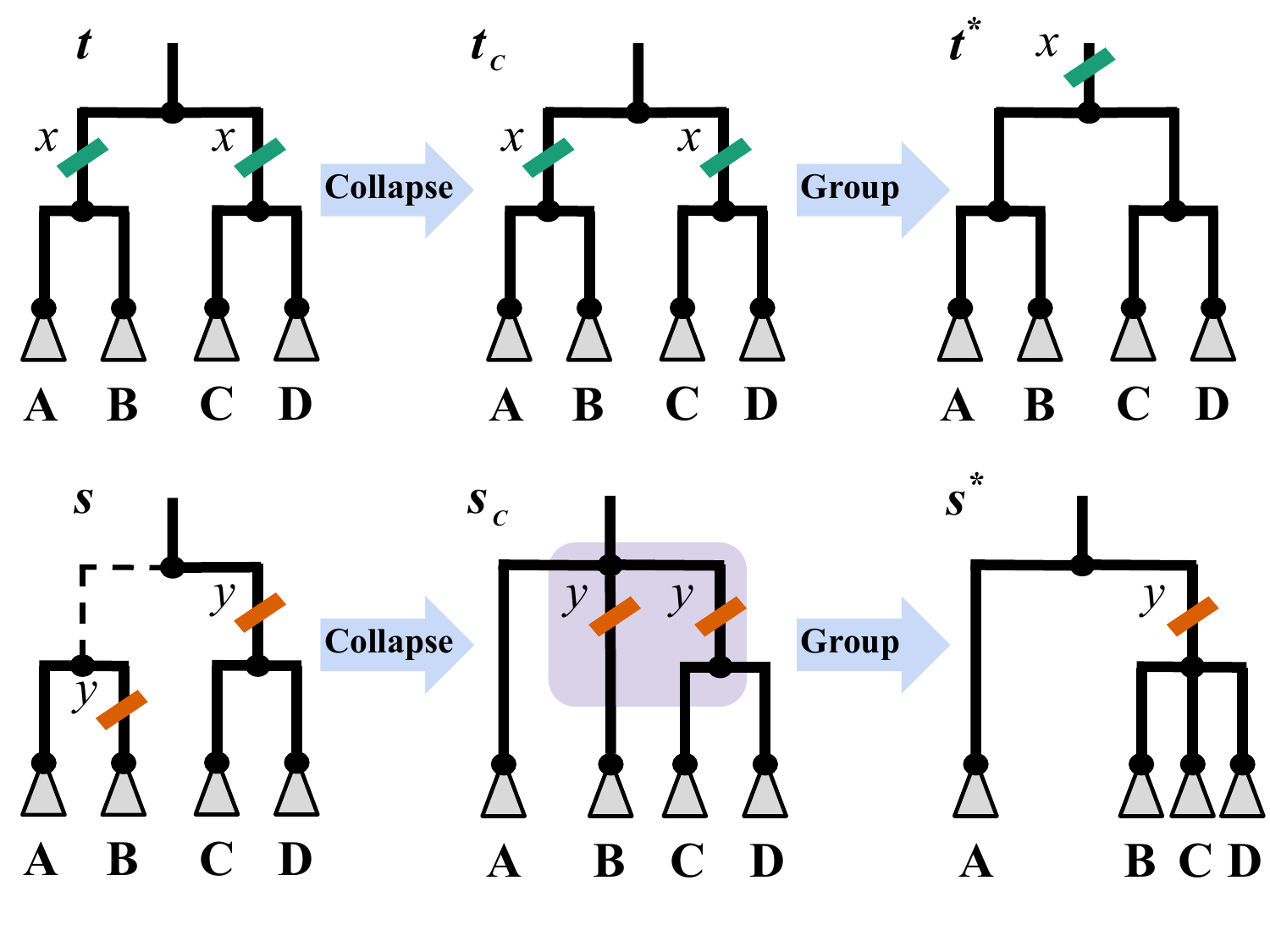}
    \caption{\
        PCMs are locally duplicated mutations.
        Two examples of PCMs occurring in simulated trees $t $ and $s $, and the corresponding most similar MP trees $t^* $ and $s^* $, are shown.
        The PCM in $s $ causes a structural difference from the MP tree $s^* $.
        A PCM is detected by examining the collapsed simulated trees $t_c $ and $s_c $ for parallel child mutations on sister branches.
    }%
    \label{fig:pcm_example}
\end{figure*}

When we simulate realistic sequence data using the methods in \autoref{sec:simulations}, we find that the generative tree (topology) $t$ is rarely MP.
Let us refer to the trees contained in the trimmed history sDAG (\autoref{sec:hdag_mp_tree_search}) as \emph{inferred-MP trees}.
Despite being subparsimonious, $t$ often shares many clades with these inferred-MP trees, and the ways in which $t$ deviates from them are through local simple structures.

The most common type of local structural difference results from mutation patterns we call parallel child mutations (PCMs).
A PCM is a subset of mutations that occur independently on two or more sister branches when a tree is collapsed to include only branches with mutations.
The idea is that a PCM is a feature of a history representing independent, convergent evolution on nearby branches.
PCMs are by definition subparsimonious, but they do not always induce subparsimonious structure.
For example, consider the top example in \autoref{fig:pcm_example}.
The subset of the true phylogenetic history $t$ is on the left and the subset of the history we infer using the MP criterion $t^*$ is on the right.
The green tick indicates a set of mutations $x$ that the taxa in subhistories A, B, C, and D all share.
To make the simulated history more parsimonious, we can replace the two mutations with a single one on the branch corresponding to the clade $A \cup B \cup C \cup D$, as in $t^*$.
However, this does not change the tree topology itself.

We are more interested in situations where PCMs create subparsimonious \emph{topology}.
For example, consider the bottom example in \autoref{fig:pcm_example}.
The subset of the true phylogenetic history $s$ is on the left and the subset of the history we infer using the MP criterion $s^*$ is on the right.
In this case, the orange tick indicates a mutation shared by B, C, and D.
When $s$ is collapsed to $s_c$ by removing the dashed branch, and B, C, and D are grouped under a single branch, the topology changes.
Specifically, $s^*$ contains the split $A \mid B, C, D$, which is incompatible with the split $A, B \mid C, D$ in $s$.
We observe that simulated trees have many topologically subparsimonious PCM-induced structures like those in $s$, but are globally quite similar to an MP tree.

\begin{figure*}[!t]
    \centering
    \includegraphics[width=\textwidth]{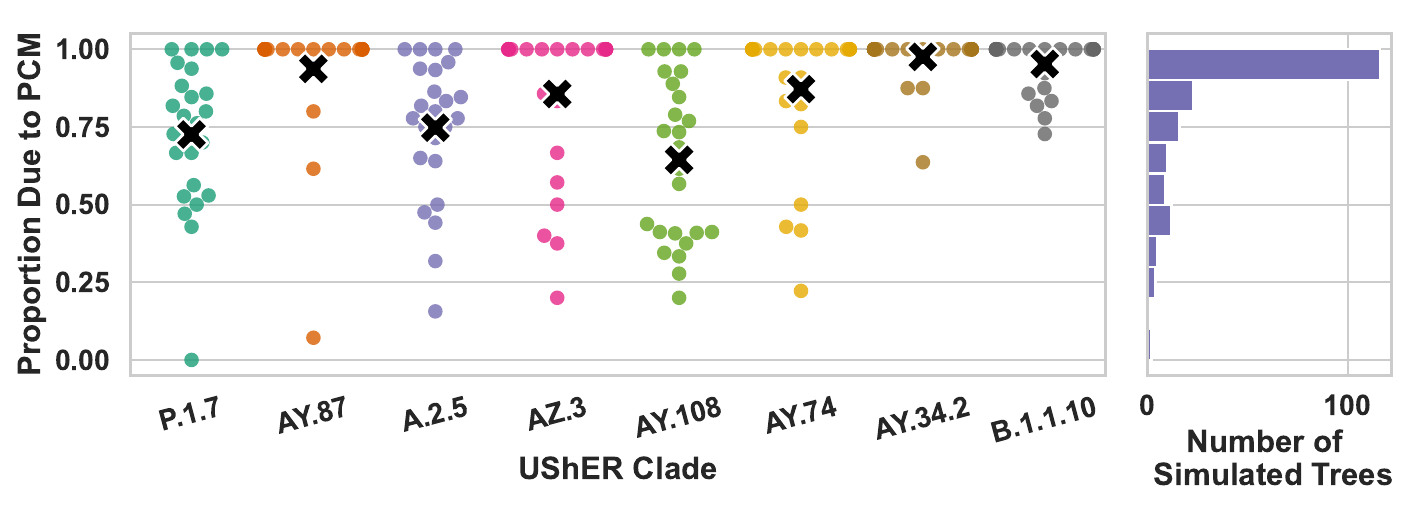}
    \caption{\
        For simulations based on each \usher{}-clade, the swarm plot on the left shows the proportion of nodes in an MP tree that are missing in the corresponding true tree because of a PCM.
        The black X denotes the average within each clade.
        Some points on the $y=1$ line are overlapping, so the histogram shows the number of simulations with a particular proportion of differing nodes that can be explained by PCMs.
    }%
    \label{fig:PCM_result}
\end{figure*}

Examining all 200 simulated trees programmatically, we find that in at least 53.5\% of the trials the \emph{entire, tree-wide} difference between simulated tree and the nearest inferred-MP tree could be explained by PCMs (see Supplementary Materials \autoref{sec:identifying_pcm_details} for further method details).
\autoref{fig:PCM_result} breaks down the proportion of differing nodes that are due to a PCM by \usher{}-clade.
The \usher{}-clade used to seed the simulation significantly affects this proportion.
For clade AY.108, this number is as low as 60.6\%, however for many clades it's over 80\%.
Across all simulations, 1693 differing nodes out of the 2270 total (74.6\%) could be explained by PCMs.
These proportions are likely underestimates, since there may exist trees of the same parsimony score that are more similar to simulated trees, which we failed to find in our MP tree search.
Later, we show how this insight can be used to estimate clade supports by modifying trees from the history sDAG to include these subparsimonious structures.

\subsection{PCM frequency increases with mutation rate}
\label{sec:pcms_vs_mutation_rate}

We hypothesized that the frequency with which a substitution is involved in a PCM increases with the mutation rate at the substitution site and target base.
After all, instances of these high-rate substitutions will happen more frequently, and thus are more likely to occur independently on sister branches in PCMs.
Getting a better understanding of this relationship could help us predict the frequency of PCMs in real examples, and potentially inform the design of phylogenetic algorithms that leverage this information.

\begin{figure}[!t]
    \centering
    \includegraphics[width=0.48\textwidth]{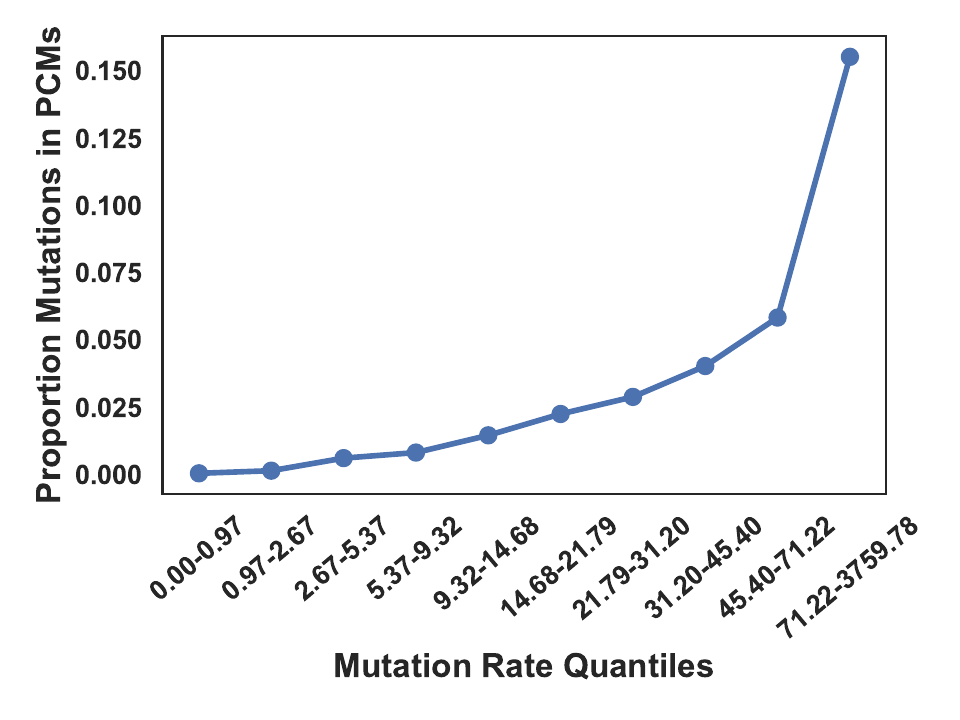}
    \caption{\
        The pointplot shows the relationship between PCM proportion and mutation rate. The x-axis shows mutation-quantile buckets that contain mutations whose rate is within a given interval.
        The intervals were selected so that each contains 10\% of all the mutations that occur across all 400 simulations.
        The y-axis is the proportion of mutations in a given bucket that were a part of a PCM.
    }%
    \label{fig:PCM_vs_mutation_rate}
\end{figure}

Indeed, when we empirically examine this relationship in our simulations, we find that there is a clear increasing relationship (\autoref{fig:PCM_vs_mutation_rate}).
To perform this analysis, we summarized all substitutions that occurred in the simulations, along with their mutation rates and whether they are part of a PCM.
Then, we selected decile intervals of mutation rate and plotted the average number of mutations within each interval.
The frequency of PCMs steadily increases with mutation rate for the first nine deciles.
No mutations with a mutability between 0 and 1 occur in PCMs, while over 5\% of the mutations with a rate in the 45 to 70 range occur in PCMs.
From the ninth to the tenth decile there is a sharp increase by almost 10\%, which we suspect is due to the fact that the last decile contains mutations with much larger mutation rates.

\subsection{Estimating clade support from MP trees}
\label{sec:support_estimation}
One useful output of many phylogenetic algorithms is a measure of confidence in particular groupings of observed taxa in the inferred tree.
This \emph{clade support} measurement provides context for an interpretation of the inferred tree, such as an observation about the number of introductions that seeded an outbreak in a region.
Often, clade support is computed by sampling many trees from a posterior distribution over trees using, e.g., Markov chain Monte Carlo (MCMC) and returning the proportion of samples that contain the given clade \cite{mrbayes2012,Drummond2007BEAST}.
In this case, a clade's support can be interpreted as the probability that that clade is contained in the true tree under the specific evolutionary model assumptions.
However, running MCMC can be prohibitively slow for large sequence alignments.
Another approach is the bootstrap, which resamples the columns of the sequence alignment and then infers a tree on each resampled alignment \cite{Felsenstein1985Bootstrap}.
This approach is less attractive in the case of densely-sampled viral data, where a given true edge may be supported by only one or a few mutations~\cite{Morel2020-mb}.
In this paper, we present two alternative support estimators that leverage the set of MP trees and can be computed quickly using the history sDAG.

To explore our observation that PCMs are the primary reason that realistic trees fail to be maximally parsimonious, we estimate clade support by sampling from two simplified distributions.
The first distribution is uniform on all inferred-MP trees, and the second is more diffuse, introducing PCMs to inferred-MP trees during the sampling process.

\subsubsection*{\bf Uniform MP posterior}
As a first attempt (to be improved upon below) at estimating clade support from collections of highly parsimonious histories, one can make the simplifying assumption that the posterior is uniform among the MP histories (Definition \ref{defn:history}).
After all, in MP inference these histories are deemed equally optimal, and we have no reason to expect any individual one to be more likely than the rest.
In that case, the support for a clade is simply the proportion of MP histories which contain that clade.

To estimate this quantity, we build a large history sDAG (\autoref{sec:hdag_mp_tree_search}), trim it to express only the most parsimonious histories (\autoref{sec:min_weight_trim}), and count the proportion of inferred-MP histories that contain a given clade.
This last step can be done efficiently with a dynamic programming algorithm in the history sDAG.
Note that this procedure only approximates the uniform MP posterior because the history sDAG will probably only contain a subset of the MP trees.
Still, to our knowledge, this is the best way to produce a large set of highly parsimonious (and likely MP) trees.

A more fundamental limitation of this support estimator is that it relies on the assumption that non-MP trees have no probability.
In our simulations, we find that the simulated tree is rarely MP.
So, under this uniformity assumption, the simulated tree will always have zero probability, which seems highly undesirable.

\subsubsection*{\bf MP-diffused posterior}
\label{sec:mp_diff}
To account for the fact that non-MP histories should be assigned non-zero probability, we  define a distribution that spreads out the probability mass on each inferred-MP history.
We achieve this by sampling a history from the history sDAG, randomly perturbing it, and counting the proportion of resulting histories that contain the clade of interest.

Our random perturbation is designed to increase the likelihood of sampling histories with PCMs to reflect the insight that the inferred-MP histories differ from the simulated tree primarily by PCMs.
To do this, we \emph{collapse} edges according to the probability that their mutations occurred independently and resolve the resulting history (for an example see $s \rightarrow s_c$ in \autoref{fig:pcm_example}).
Specifically, for the set of mutations $M$ on a given edge $e$ we consider the probability that those mutations occurred independently.
We sample a number of additional mutations $n \sim \text{Binomial}(|Ch(e)|, p)$, where $Ch(e)$ is the set of child edges of $e$, $p = \prod_{m \in M} p_m$, and $p_m$ is the probability that mutation $m$ occurred.
We randomly select $n$ edges among $Ch(e)$ to collapse, and do so by removing them from their parent node, attaching them to the parent node of $e$, and adding $M$ to their list of mutations.
We repeat this probabilistic collapse step for each edge of the history in a pre-order traversal, and return a uniformly random fully resolved history that is compatible with the collapsed version.
The result is a history sampled from a distribution that has non-zero probability of sampling histories with PCMs.

If there are no mutations on the edge, we set the probability of collapse to 1 to reflect the fact that there is no phylogenetic support for that edge being present in the generative tree.
For edges with mutations, we find that setting the probability of each mutation to $p_m = 0.02$ works well on our simulations.
We observed that the mean of distribution of parsimony scores of histories sampled with this setting of $p_m$ was near the parsimony score of the simulated history for many simulated datasets.
Of course, this cannot be done without knowledge of the true number of mutations.
We expected that estimating $p_m$ from features of the dataset like parsimony diversity would yield similar support estimation results, and indeed see this in a preliminary implementation (results not shown).

\subsubsection*{\bf Fast diffused support sampling}
The two-step sampling procedure to estimate the MP-diffused posterior is very sample-inefficient.
Indeed, we compute this estimator by sampling histories from the history sDAG, injecting PCMs, resolving multifurcating nodes uniformly, and adding one unit to a counter associated to each clade observed in the resulting tree.
However, it is common for MP trees to include large multifurcations, where the data do not inform a choice between any of the many possible binary resolutions.
Such a large multifurcation may need to be sampled, and resolved randomly, very many times for the resulting clade supports to approximate the true support, since many clades occur in only a tiny fraction of possible binary resolutions.

We can reduce the number of samples required, and sample more clades with a single MP tree chosen from the history sDAG, by avoiding taking a single choice of binary resolution.
Instead, we can compute the fraction of binary resolutions containing each clade that could possibly result from a resolution of a multifurcating node.
Instead of adding $1 $ to the counter for each resulting clade, we can add this fraction.

A multifurcating node with $n $ children has $b(n) = (2n - 3)!! $ possible binary resolutions.
To compute how many of these resolutions contains a particular monophyletic grouping of $s < n$ children of the node, we need only multiply the number of binary arrangements of those $n $ children, and the number of binary arrangements relating the remaining $n - s $ leaves and the clade of size $s $.
Therefore, there are $b(s) \times b(n - s + 1) $ binary resolutions of a multifurcating node with $n $ children, containing each possible child clade of size $s $.
We can convert this to a fraction of possible binary resolutions by normalizing by $b(n) $.

Since the set of clades which can result from resolving a multifurcation of size $n $ scales as the powerset of $n $ elements, it is impractical to iterate through each possible clade for large multifurcations.
Instead, we ignore any clades whose frequency among all possible binary resolutions of the node would be less than a threshold of our choosing.
This adjustment to our clade sampling method allows us to make more accurate estimates of clade supports with fewer samples of MP trees from the history sDAG.

\subsection{Support evaluation}

\begin{figure*}[!t]
    \centering
    \includegraphics[width=0.75\textwidth]{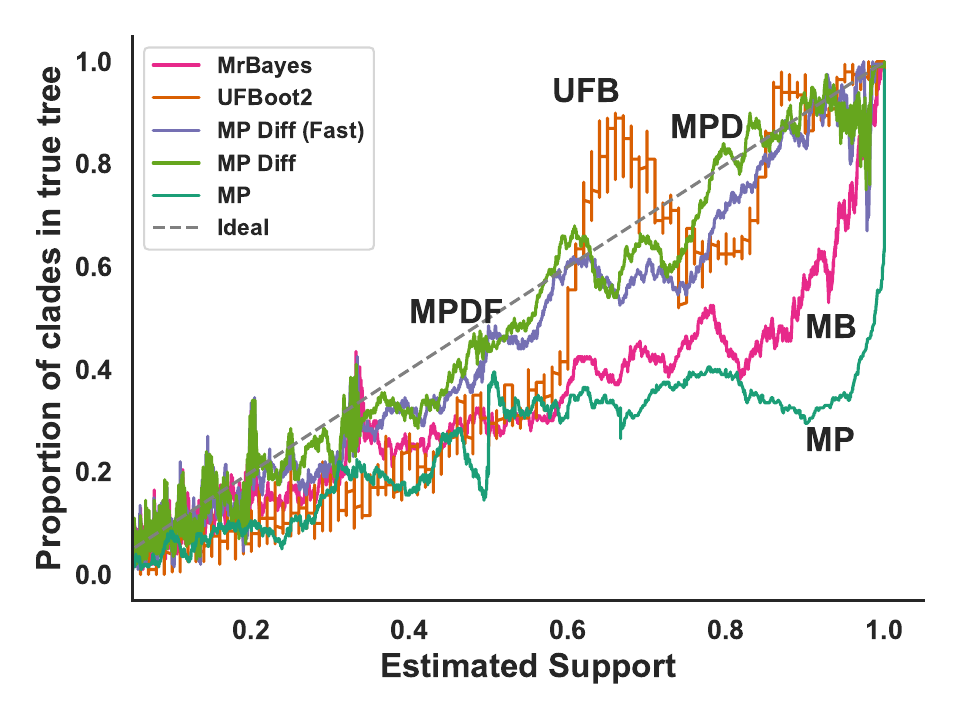}
    \caption{\
    Calibration curves for each support estimation method: MrBayes, UFBoot2, MP, MP-Diffused, and MP-Diffused (Fast). Direct labels are MB, UFB, MP, MPD, and MPDF, respectively.
    The x-axis is the estimated support value, and the y-axis is the proportion of clades in the 200 length window that were in the true tree.
    For ease of plotting, we omit points with an estimated support in the range $(0, 0.05)$, which does not affect the trend.
    }%
    \label{fig:coverage_plot}
\end{figure*}
We compare our MP-based support estimators to two traditional methods for support estimation: Bayesian posterior, implemented as MrBayes \cite{mrbayes2012},  and Bootstrap, implemented as Ultrafast Bootstrap (UFBoot2) \cite{Hoang2017UFBoot2}.
For MrBayes, we used $10^{8}$ total iterations, with $7\cdot 10^{7}$ iterations of burnin, and sampled a tree once every $1000$ iterations. We also used the most general GTR model with four categories of Gamma distributed rate variation.
For UFBoot2, we used 1000 bootstrap replicates and allowed nearest neighbor interchange operations to further optimize bootstrapped trees \cite{Simmons2014DivergentMV}.

For each method, we plot a calibration curve by first computing the support values for data generated by the simulation procedure described in \autoref{sec:simulations}.
We then sort the clades (across all simulated datasets) by their estimated support value.
For each estimated support value, we compute the empirical probability for that clade given its support value as the proportion of clades in a 200 length sliding window that are present in the generative tree. The calibration plot in \autoref{fig:coverage_plot} summarizes these results.

In \autoref{fig:coverage_plot}, overestimation can be seen when the plotted curve is far to the right of the ideal $y=x $ line.
This indicates that the estimated support is greater than true support.
The default MP support estimator consistently overestimates the supports values (i.e., for any given clade the estimated support is \emph{higher} than its empirical support).
MrBayes performs slightly better, but shows similar overestimation for estimated values in the range of 0.4 to 0.95.
UFBoot2 also overestimates, but its curve is significantly closer to ideal than MrBayes or default MP support, especially for estimated supports in the range of 0.6 to 0.9.
The calibration curves for both versions of MP-Diffused are closest to ideal for almost all predicted support values.

In addition to providing a potentially useful way of accurately estimating clade supports on densely sampled data, the success of the MP-Diffused clade support method provides further evidence that we can closely approximate the true distribution of trees conditioned on the data by adding PCMs to MP trees.

\section{Conclusions}
\label{sec:conclusions}

\noindent In summary, we have found that subparsimonious structures in simulated densely sampled phylogenies are often the result of a simple, local mutation duplication which we call a PCM.
This observation allows us to estimate clade support by perturbing MP trees to include these subparsimonious structures.
We find that this is a useful clade support estimator for simulations that are calibrated to mimic real densely-sampled data.

\subsection{Parsimony-suboptimal structures in simulations}
Subparsimoniousness is a classical concept and many metrics have been proposed to measure it \cite{Brooks1986AMO,Farris1989THERI}.
One widely used approach for measuring subparsimoniousness is Consistency Index (CI), which measures the amount of homoplasy in a given tree by taking the ratio $m/s$, where $m$ is the minimum number of character state changes to explain the observed sequences (i.e., the MP score) and $s$ is the observed number of character state changes that occurred.
A CI score of 1 means that the sequences evolved without homoplasy, and as the number of homoplasies goes up the CI goes down.
This metric gives a sense of how close an evolutionary history is to being MP, but it does not shed any light on which substructures are subparsimonious and what mutation patterns give rise to them.
In this work, we focus on these latter questions in a specific, yet highly relevant, parameter regime.

Despite the fact that homoplasy is a realistic feature of evolutionary histories, the MP criterion remains useful, particularly in the densely sampled regime where datasets are large and branch lengths are short.
Theoretical work has also shown that the ML tree will be the same as the MP tree under simple models of evolution if the branch lengths are assumed to be very small \cite{felsenstein1973maximum}.
More recent results show that in the near-perfect regime where tree length is small, it is likely that the splits in the (collapsed) MP tree are correct \cite{Wertheim2022-bm}.
In our simulations, we confirm that the true tree can frequently be made maximally parsimonious by making local edits around the subparsimonious splits (e.g., collapsing short edges and grouping them under a common mutation).

We go further by showing that these subparsimonious splits are often PCMs~(\autoref{sec:identifying_pcms}), and propose a method~(\autoref{sec:mp_diff}) that uses this insight to perturb MP trees in more realistic ways.
Although this approach produces fairly accurate estimates of clade support, the way in which it perturbs MP trees is very simple.
When collapsing an edge, it only considers the number of mutations on that branch and ignores the types of mutations that occur on it.
One line of future work could examine ways to make this form of perturbation even more realistic by, for example, taking the mutability of the mutations on a branch into account.
Our insight that PCMs are much more frequent among mutations with high mutability~(\autoref{sec:pcms_vs_mutation_rate}) indicates that incorporating this information into the MP tree perturbation could further improve tree sampling.

\subsection{Comparison to other methods}
Support estimates from diffused MP-sampling are well calibrated compared to the other approaches, and show the least overestimation.
All five methods exhibit overestimation: the naive MP clade support estimate is the most overconfident, followed by MrBayes and UFBoot2, and the two MP-diffused estimates appear to be most accurate.

For the naive MP clade support and MrBayes methods, the overestimation of support is correlated with the assumptions each method makes about parsimony scores of realistic trees.
The naive MP clade support method is based on the certainly incorrect assumption that all plausible trees explaining the data are MP.
It will therefore place too much probability on the clades that are among MP trees, and overestimate their support.

One could argue for the use of Bayesian methods to infer phylogenetic relationships between densely sampled data since they allow for subparsimonious structures.
However, we observe that a straightforward application of MrBayes, using the most general model settings available, samples trees that are also substantially more parsimonious than the simulated tree.
In fact, if we compute the best possible parsimony score of any tree on the simulated data, we observe that the distribution of parsimony scores in MrBayes samples is often concentrated very near this best MP score.
\autoref{fig:parsimony_posterior} in the Appendix shows an example of this.
Since a matching parsimony score is a necessary condition for two trees to have the same topology, this makes sampling the true tree with MrBayes extremely unlikely.
In fairness to MrBayes, this behavior is certainly the result of a model misspecification, since the most general model available in MrBayes is not expressive enough to match our simulation model.
However, since our simulations are designed to mimic real densely sampled data, we can expect a similar model misspecification, and corresponding overconfidence in clade support estimates, on real data.

Another challenge with Bayesian methods is that they are computationally intensive and cannot scale to extremely large sets of densely sampled data, such as that available for \sarscov{}.
Recent work leverages parsimony to inform proposals that speed up MCMC chain convergence and improve mixing time \cite{zhang2020using}.
We hope our results about subparsimonious structure, and our ability to sample realistic trees from a collection of MP trees, might motivate new proposal distributions for a Bayesian phylogenetic program.

UFBoot2 computes clade support estimates much faster than MrBayes.
Also, the authors show empirically that UFBoot2 can yield high-accuracy estimates of clade-membership probability even in the presence of model violations and thus be interpreted as probabilities.
However, \autoref{fig:coverage_plot} shows that UFBoot2 clade support estimates do not match observed clade supports in our simulations.
Moreover, unlike the other methods, the support values from UFBoot2 are not always overestimates.
For support ranging from 0.6 to 1, the support values oscillate between overestimation and underestimation.
Since the support values produced by UFBoot2 use ML trees on \emph{bootstrapped sequences}, it is not clear how this estimate is related to a distribution of phylogenetic trees on the \emph{real sequences}, and makes it difficult to discern the source of these inaccuracies.

\subsection{Importance of the history sDAG}
Producing clade support estimates from MP trees would be extremely difficult without the history sDAG.
Because parsimony diversity is often high on densely sampled, \sarscov-like data, any clade support estimate based on a single MP tree would likely be quite sensitive to \emph{which} MP tree is chosen.
The vast diversity of MP trees necessitates the use of an efficient data structure for searching for, storing, and sampling MP trees.
The history sDAG is therefore essential for tuning simulation parameters using parsimony diversity and sampling from the large collection of MP trees, which is necessary to estimate clade support.


\section*{Competing Interests}
\noindent The authors have no competing interests to declare.

\section*{Author Contributions}
\noindent WHS, WD, and FAM conceptualized the work.
WHS simulated data and organized PCM and support evaluation experiments.
WHS and WD developed the support algorithm.
WD and CW developed the minimum trimming algorithm.
WD, DHR, OM, and CY developed Larch, the C++ version of the history sDAG.
WHS, WD, and FAM prepared the original draft of the manuscript.
WHS, WD, FAM, and MAS edited the final draft of the manuscript.

%

\section*{Acknowledgements}
\noindent We thank Dr.~Nicola De Maio for helpful correspondence about the ideas in this work, especially in interpreting the effects of simulation parameters on generated data.
We would also like to thank Cheng Ye for correspondence about adapting the matOptimize software, as well as Russ Corbett-Detig and Yatish Turakhia.

This work was supported through US National Institutes of Health grant AI162611.
Scientific Computing Infrastructure at Fred Hutch was funded by ORIP grant S10OD028685.
Dr.\ Matsen is an Investigator of the Howard Hughes Medical Institute.

\section*{Availability of Data and Materials}
\noindent The \sarscov{} data used to produce clade simulations in the Simulations section was read from the public \sarscov{} tree distributed by the \usher{} team at \url{http://hgdownload.soe.ucsc.edu/goldenPath/wuhCor1/UShER_SARS-CoV-2/}.
This data originates from GenBank \cite{GenBankCite} at \url{https://www.ncbi.nlm.nih.gov}, COG-UK \cite{COGUKCite} at \url{https://www.cogconsortium.uk/tools-analysis/public-data-analysis-2/}, and the China National Center for Bioinformation \cite{CNCB1, CNCB2, CNBC3, CNBC4} at \url{https://bigd.big.ac.cn/ncov/release_genome}.

\section*{Code Availability}
\noindent The history sDAG data structure described in \cite{dumm2023representing}, as well as various algorithms described in this paper and in future work, are implemented in the open source Python package \historydag, which is available at \url{https://github.com/matsengrp/historydag}.
All code necessary to reproduce the \sarscov{} simulations and analysis is available at \url{https://github.com/matsengrp/hdag-benchmark}.

\bibliographystyle{IEEEtran}
\bibliography{IEEEabrv, main}

\begin{thebibliography}{10}
\providecommand{\url}[1]{#1}
\csname url@samestyle\endcsname
\providecommand{\newblock}{\relax}
\providecommand{\bibinfo}[2]{#2}
\providecommand{\BIBentrySTDinterwordspacing}{\spaceskip=0pt\relax}
\providecommand{\BIBentryALTinterwordstretchfactor}{4}
\providecommand{\BIBentryALTinterwordspacing}{\spaceskip=\fontdimen2\font plus
\BIBentryALTinterwordstretchfactor\fontdimen3\font minus \fontdimen4\font\relax}
\providecommand{\BIBforeignlanguage}[2]{{%
\expandafter\ifx\csname l@#1\endcsname\relax
\typeout{** WARNING: IEEEtran.bst: No hyphenation pattern has been}%
\typeout{** loaded for the language `#1'. Using the pattern for}%
\typeout{** the default language instead.}%
\else
\language=\csname l@#1\endcsname
\fi
#2}}
\providecommand{\BIBdecl}{\relax}
\BIBdecl

\bibitem{Wertheim2022-bm}
\BIBentryALTinterwordspacing
J.~O. Wertheim, M.~Steel, and M.~J. Sanderson, ``\BIBforeignlanguage{en}{{Accuracy in {Near-Perfect} Virus Phylogenies}},'' \emph{\BIBforeignlanguage{en}{Syst. Biol.}}, vol.~71, no.~2, pp. 426--438, Feb. 2022. [Online]. Available: \url{http://dx.doi.org/10.1093/sysbio/syab069}
\BIBentrySTDinterwordspacing

\bibitem{Kramer2023-sp}
\BIBentryALTinterwordspacing
A.~M. Kramer, B.~Thornlow, C.~Ye, N.~De~Maio, J.~McBroome, A.~S. Hinrichs, R.~Lanfear, Y.~Turakhia, and R.~Corbett-Detig, ``\BIBforeignlanguage{en}{{Online Phylogenetics with matOptimize Produces Equivalent Trees and is Dramatically More Efficient for Large {SARS-CoV-2} Phylogenies than de novo and {Maximum-Likelihood} Implementations}},'' \emph{\BIBforeignlanguage{en}{Syst. Biol.}}, May 2023. [Online]. Available: \url{http://dx.doi.org/10.1093/sysbio/syad031}
\BIBentrySTDinterwordspacing

\bibitem{USHER2021}
\BIBentryALTinterwordspacing
Y.~Turakhia, B.~Thornlow, A.~Hinrichs, N.~D. Maio, L.~Gozashti, R.~Lanfear, D.~Haussler, and R.~B. Corbett-Detig, ``{{Ultrafast Sample placement on Existing tRees} ({UShER}) enables real-time phylogenetics for the {SARS-CoV-2} pandemic},'' \emph{Nature Genetics}, vol.~53, pp. 809 -- 816, 2021. [Online]. Available: \url{https://api.semanticscholar.org/CorpusID:234360474}
\BIBentrySTDinterwordspacing

\bibitem{larch}
M.~Barker, O.~Milanov, W.~Dumm, D.~Rich, Y.~Turakhia, and F.~A.~M. IV, ``Larch: mapping the parsimony-optimal landscape of trees for directed exploration,'' 2025, manuscript under review.

\bibitem{dumm2023representing}
\BIBentryALTinterwordspacing
W.~Dumm, M.~Barker, W.~Howard-Snyder, W.~S. DeWitt~III, and F.~A. Matsen~IV, ``{Representing and extending ensembles of parsimonious evolutionary histories with a directed acyclic graph},'' \emph{Journal of Mathematical Biology}, vol.~87, no.~5, p.~75, 2023. [Online]. Available: \url{https://doi.org/10.1007/s00285-023-02006-3}
\BIBentrySTDinterwordspacing

\bibitem{haag2022}
\BIBentryALTinterwordspacing
J.~Haag, D.~Höhler, B.~Bettisworth, and A.~Stamatakis, ``{From Easy to Hopeless—Predicting the Difficulty of Phylogenetic Analyses},'' \emph{Molecular Biology and Evolution}, vol.~39, no.~12, 11 2022. [Online]. Available: \url{https://doi.org/10.1093/molbev/msac254}
\BIBentrySTDinterwordspacing

\bibitem{MAPLE2022}
\BIBentryALTinterwordspacing
N.~D. Maio, P.~Kalaghatgi, Y.~Turakhia, R.~B. Corbett-Detig, B.~Q. Minh, and N.~Goldman, ``{Maximum likelihood pandemic-scale phylogenetics},'' \emph{bioRxiv}, 2022. [Online]. Available: \url{https://api.semanticscholar.org/CorpusID:247631660}
\BIBentrySTDinterwordspacing

\bibitem{mcbroome2021}
\BIBentryALTinterwordspacing
J.~McBroome, B.~Thornlow, A.~S. Hinrichs, A.~Kramer, N.~De~Maio, N.~Goldman, D.~Haussler, R.~Corbett-Detig, and Y.~Turakhia, ``{A Daily-Updated Database and Tools for Comprehensive {SARS-CoV-2} Mutation-Annotated Trees},'' \emph{Molecular Biology and Evolution}, vol.~38, no.~12, pp. 5819--5824, 09 2021. [Online]. Available: \url{https://doi.org/10.1093/molbev/msab264}
\BIBentrySTDinterwordspacing

\bibitem{2022phastsim}
N.~De~Maio, W.~Boulton, L.~Weilguny, C.~R. Walker, Y.~Turakhia, R.~Corbett-Detig, and N.~Goldman, ``{phastSim: efficient simulation of sequence evolution for pandemic-scale datasets},'' \emph{PLoS computational biology}, vol.~18, no.~4, p. e1010056, 2022.

\bibitem{demaio2021}
\BIBentryALTinterwordspacing
N.~De~Maio, C.~R. Walker, Y.~Turakhia, R.~Lanfear, R.~Corbett-Detig, and N.~Goldman, ``{Mutation Rates and Selection on Synonymous Mutations in {SARS-CoV-2}},'' \emph{Genome Biology and Evolution}, vol.~13, no.~5, p. evab087, 04 2021. [Online]. Available: \url{https://doi.org/10.1093/gbe/evab087}
\BIBentrySTDinterwordspacing

\bibitem{mrbayes2012}
\BIBentryALTinterwordspacing
F.~Ronquist, M.~Teslenko, P.~van~der Mark, D.~L. Ayres, A.~Darling, S.~Höhna, B.~Larget, L.~Liu, M.~A. Suchard, and J.~P. Huelsenbeck, ``{{MrBayes} 3.2: Efficient Bayesian Phylogenetic Inference and Model Choice Across a Large Model Space},'' \emph{Systematic Biology}, vol.~61, no.~3, pp. 539--542, 02 2012. [Online]. Available: \url{https://doi.org/10.1093/sysbio/sys029}
\BIBentrySTDinterwordspacing

\bibitem{Drummond2007BEAST}
\BIBentryALTinterwordspacing
A.~J. Drummond and A.~Rambaut, ``{BEAST: Bayesian evolutionary analysis by sampling trees},'' \emph{BMC Evolutionary Biology}, vol.~7, no.~1, p. 214, 2007. [Online]. Available: \url{https://doi.org/10.1186/1471-2148-7-214}
\BIBentrySTDinterwordspacing

\bibitem{Felsenstein1985Bootstrap}
\BIBentryALTinterwordspacing
J.~Felsenstein, ``{Confidence Limits on Phylogenies: An Approach Using the Bootstrap},'' \emph{Evolution}, vol.~39, no.~4, pp. 783--791, 1985. [Online]. Available: \url{http://www.jstor.org/stable/2408678}
\BIBentrySTDinterwordspacing

\bibitem{Morel2020-mb}
\BIBentryALTinterwordspacing
B.~Morel, P.~Barbera, L.~Czech, B.~Bettisworth, L.~H{\"u}bner, S.~Lutteropp, D.~Serdari, E.-G. Kostaki, I.~Mamais, A.~M. Kozlov, P.~Pavlidis, D.~Paraskevis, and A.~Stamatakis, ``\BIBforeignlanguage{en}{{Phylogenetic analysis of {SARS-CoV-2} data is difficult}},'' \emph{\BIBforeignlanguage{en}{Mol. Biol. Evol.}}, Dec. 2020. [Online]. Available: \url{http://dx.doi.org/10.1093/molbev/msaa314}
\BIBentrySTDinterwordspacing

\bibitem{Hoang2017UFBoot2}
\BIBentryALTinterwordspacing
D.~T. Hoang, O.~Chernomor, A.~von Haeseler, B.~Q. Minh, and L.~S. Vinh, ``{UFBoot2: Improving the Ultrafast Bootstrap Approximation},'' \emph{Molecular Biology and Evolution}, vol.~35, pp. 518 -- 522, 2017. [Online]. Available: \url{https://api.semanticscholar.org/CorpusID:4359869}
\BIBentrySTDinterwordspacing

\bibitem{Simmons2014DivergentMV}
\BIBentryALTinterwordspacing
M.~P. Simmons and A.~P. Norton, ``{Divergent maximum-likelihood-branch-support values for polytomies},'' \emph{Molecular phylogenetics and evolution}, vol.~73, pp. 87--96, 2014. [Online]. Available: \url{https://api.semanticscholar.org/CorpusID:205838669}
\BIBentrySTDinterwordspacing

\bibitem{Brooks1986AMO}
\BIBentryALTinterwordspacing
D.~R. Brooks and R.~T. O'Grady, ``{A Measure of the Information Content of Phylogenetic Trees, and its use as an Optimality Criterion},'' \emph{Systematic Biology}, vol.~35, pp. 571--581, 1986. [Online]. Available: \url{https://api.semanticscholar.org/CorpusID:121169782}
\BIBentrySTDinterwordspacing

\bibitem{Farris1989THERI}
\BIBentryALTinterwordspacing
J.~S. Farris, ``{The Retention Index and Rescaled Consistency Index},'' \emph{Cladistics}, vol.~5, 1989. [Online]. Available: \url{https://api.semanticscholar.org/CorpusID:84287895}
\BIBentrySTDinterwordspacing

\bibitem{felsenstein1973maximum}
\BIBentryALTinterwordspacing
J.~Felsenstein, ``{Maximum Likelihood and Minimum-Steps Methods for Estimating Evolutionary Trees from Data on Discrete Characters},'' \emph{Systematic Biology}, vol.~22, no.~3, pp. 240--249, 09 1973. [Online]. Available: \url{https://doi.org/10.1093/sysbio/22.3.240}
\BIBentrySTDinterwordspacing

\bibitem{zhang2020using}
\BIBentryALTinterwordspacing
C.~Zhang, J.~P. Huelsenbeck, and F.~Ronquist, ``{Using Parsimony-Guided Tree Proposals to Accelerate Convergence in Bayesian Phylogenetic Inference},'' \emph{Systematic Biology}, vol.~69, no.~5, pp. 1016--1032, 01 2020. [Online]. Available: \url{https://doi.org/10.1093/sysbio/syaa002}
\BIBentrySTDinterwordspacing

\bibitem{GenBankCite}
E.~L. Hatcher, S.~A. Zhdanov, Y.~Bao, O.~Blinkova, E.~P. Nawrocki, Y.~Ostapchuck, A.~A. Schäffer, and J.~R. Brister, ``{Virus Variation Resource - improved response to emergent viral outbreaks},'' \emph{Nucleic Acids Research}, vol.~45, no.~D1, pp. D482--D490, 11 2016.

\bibitem{COGUKCite}
S.~M. Nicholls, R.~Poplawski, M.~J. Bull, A.~Underwood, M.~Chapman, K.~Abu-Dahab, B.~Taylor, B.~Jackson, S.~Rey, R.~Amato, R.~Livett, S.~Gon{\c c}alves, E.~M. Harrison, S.~J. Peacock, D.~M. Aanensen, A.~Rambaut, T.~R. Connor, and N.~J. Loman, ``{{MAJORA}: {Continuous} integration supporting decentralised sequencing for {SARS-CoV-2} genomic surveillance},'' \emph{bioRxiv}, 2020.

\bibitem{CNCB1}
S.~Song, L.~Ma, D.~Zou, D.~Tian, C.~Li, J.~Zhu, M.~Chen, A.~Wang, Y.~Ma, M.~Li, X.~Teng, Y.~Cui, G.~Duan, M.~Zhang, T.~Jin, C.~Shi, Z.~Du, Y.~Zhang, C.~Liu, R.~Li, J.~Zeng, L.~Hao, S.~Jiang, H.~Chen, D.~Han, J.~Xiao, Z.~Zhang, W.~Zhao, Y.~Xue, and Y.~Bao, ``{The Global Landscape of {SARS-CoV-2} Genomes, Variants, and Haplotypes in {2019nCoVR}},'' \emph{Genomics, Proteomics \& Bioinformatics}, vol.~18, no.~6, pp. 749--759, 2020.

\bibitem{CNCB2}
W.-M. Zhao, S.-H. Song, M.-L. Chen, D.~Zou, L.-N. Ma, Y.-K. Ma, R.-J. Li, L.-L. Hao, C.-P. Li, D.-M. Tian \emph{et~al.}, ``{The 2019 novel coronavirus resource},'' \emph{Yi chuan= Hereditas}, vol.~42, no.~2, pp. 212--221, 2020.

\bibitem{CNBC3}
Z.~Gong, J.-W. Zhu, C.-P. Li, S.~Jiang, L.-N. Ma, B.-X. Tang, D.~Zou, M.-L. Chen, Y.-B. Sun, S.-H. Song \emph{et~al.}, ``{An online coronavirus analysis platform from the National Genomics Data Center},'' \emph{Zoological research}, vol.~41, no.~6, p. 705, 2020.

\bibitem{CNBC4}
D.~Yu, X.~Yang, B.~Tang, Y.-H. Pan, J.~Yang, G.~Duan, J.~Zhu, Z.-Q. Hao, H.~Mu, L.~Dai, W.~Hu, M.~Zhang, Y.~Cui, T.~Jin, C.-P. Li, L.~Ma, L.~translation team, X.~Su, G.~Zhang, W.~Zhao, and H.~Li, ``{Coronavirus {GenBrowser} for monitoring the transmission and evolution of {SARS-CoV-2}},'' \emph{Briefings in Bioinformatics}, vol.~23, no.~2, 01 2022, bbab583.

\end{thebibliography}

\begin{IEEEbiography}[{\includegraphics[width=1in,height=1.25in,clip]{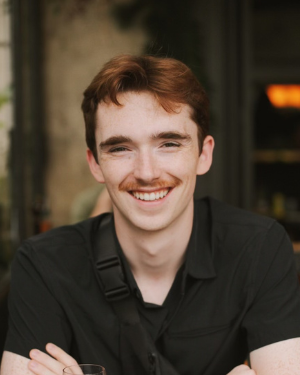}}]{William Howard-Snyder}
    received his B.S. in Computer Science from the University of Washington, Seattle, WA in 2022, and received his M.S. from the same institution in 2024.
    He is currently pursuing a Ph.D. in Computer Science at Princeton University.
    He completed this research in the Matsen Lab while earning his M.S.
    His research interests include developing methods to better understand stochastic processes in biology, such as viral evolution and mammalian embryogenesis.
\end{IEEEbiography}%

\begin{IEEEbiography}[{\includegraphics[width=1in,height=1.25in,clip]{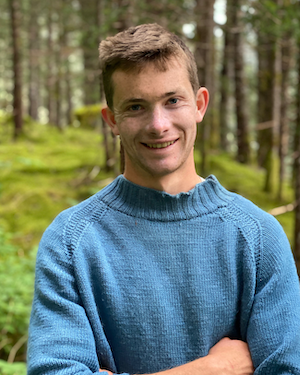}}]{Will Dumm}
    received his B.S. in Mathematics at Montana State University, Bozeman, MT in 2019, and received his M.S. from the same university in 2020.
    Since 2021, he works as a programmer in the Matsen Lab at Fred Hutchinson Cancer Center, Seattle, WA.
    There, he contributes to projects involving phylogenetic inference and models for antibody evolution.
\end{IEEEbiography}%

\begin{IEEEbiography}[{\includegraphics[width=1in,height=1.25in,clip]{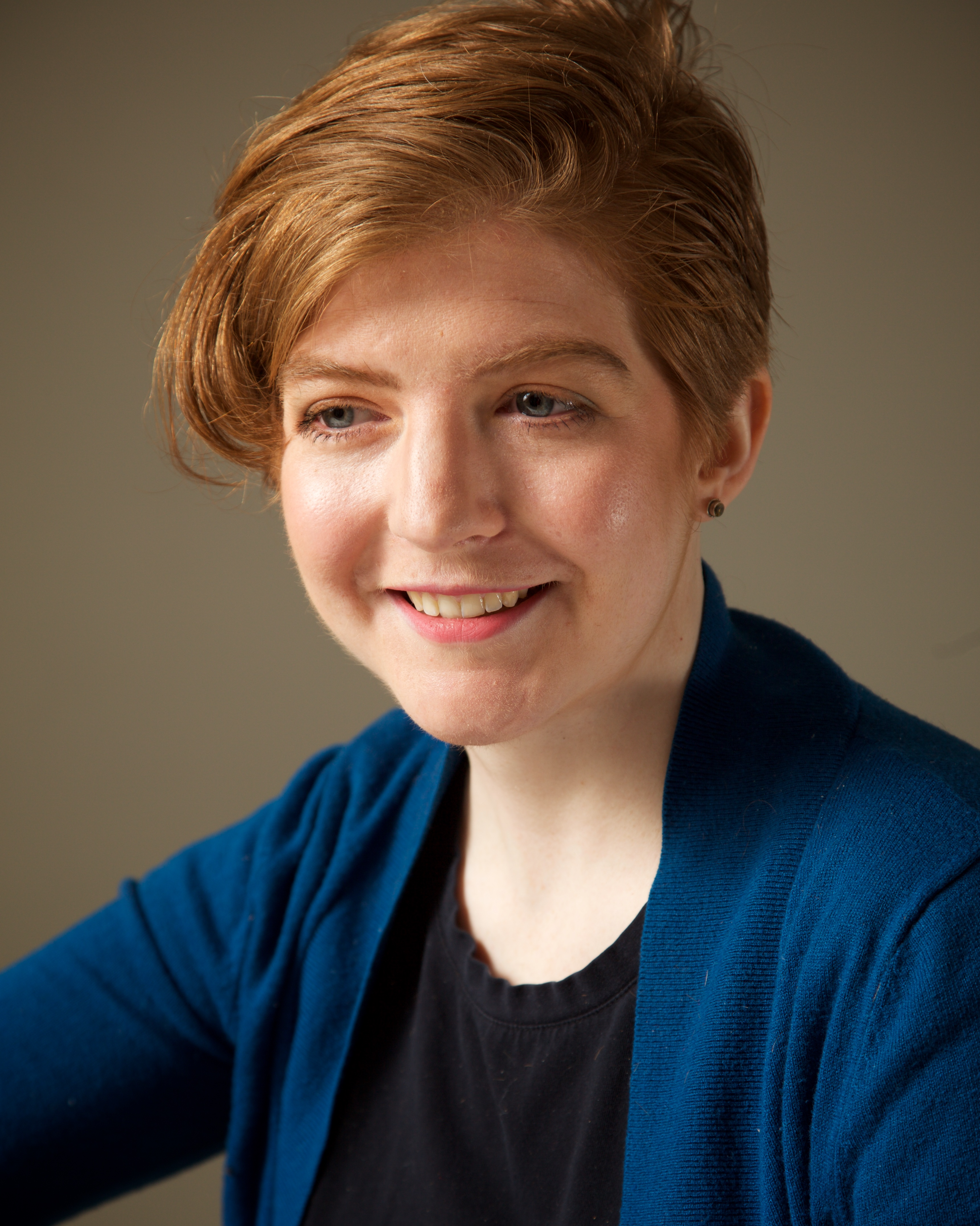}}]{Mary Barker}
    received a Bachelor's degree in Mathematics from Covenant College, Lookout Mountain, GA, and M.S. degrees in Computational Engineering, Data Mining, and Mathematics from the University of Chattanooga in Tennessee, TN, Tarleton State University, TX, and Washington University in St. Louis, MO, respectively.
    She received her Ph.D. in Mathematics from Washington University in St. Louis in 2022, and is currently working as a postdoctoral fellow in the Matsen Lab at Fred Hutchinson Cancer Center, Seattle, WA.
    Her research interests span a wide range in computational mathematics including large-scale phylogenetic reconstructions, nonconforming methods in finite element exterior calculus and quivers and their moduli spaces.
\end{IEEEbiography}%

\begin{IEEEbiography}[{\includegraphics[width=1in,height=1.25in,clip]{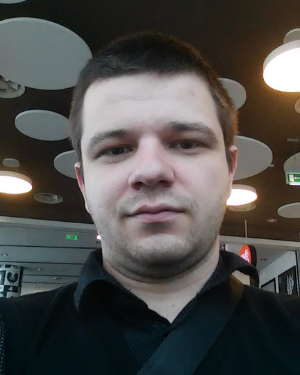}}]{Ognian Milanov}
    is working as a C++ programmer in the Matsen Lab at Fred Hutchinson Cancer Center, Seattle, WA.
    He is a software engineer with 17 years of experience in C++ development, consulting various organizations in US and EU.
    His interests involve designing and implementing high-performance, data intensive and real time applications.
\end{IEEEbiography}%

\begin{IEEEbiography}[{\includegraphics[width=1in,height=1.25in,clip]{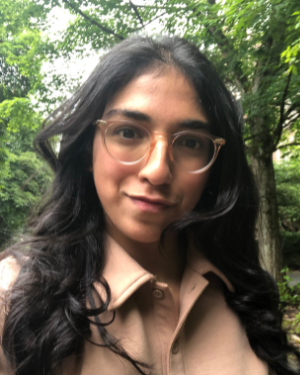}}]{Claris Winston}
    received her B.S. degree in Computer Science in 2024 from the University of Washington, Seattle, and is currently working towards her M.S. in Computer Science at the same school. She worked in the Matsen Lab at Fred Hutchinson Cancer Center, Seattle, WA, during her undergrad. Her work has revolved around applying AI and computational methods in biotechnology and accessibility domains.
\end{IEEEbiography}

\begin{IEEEbiography}[{\includegraphics[width=1in,height=1.25in,clip]{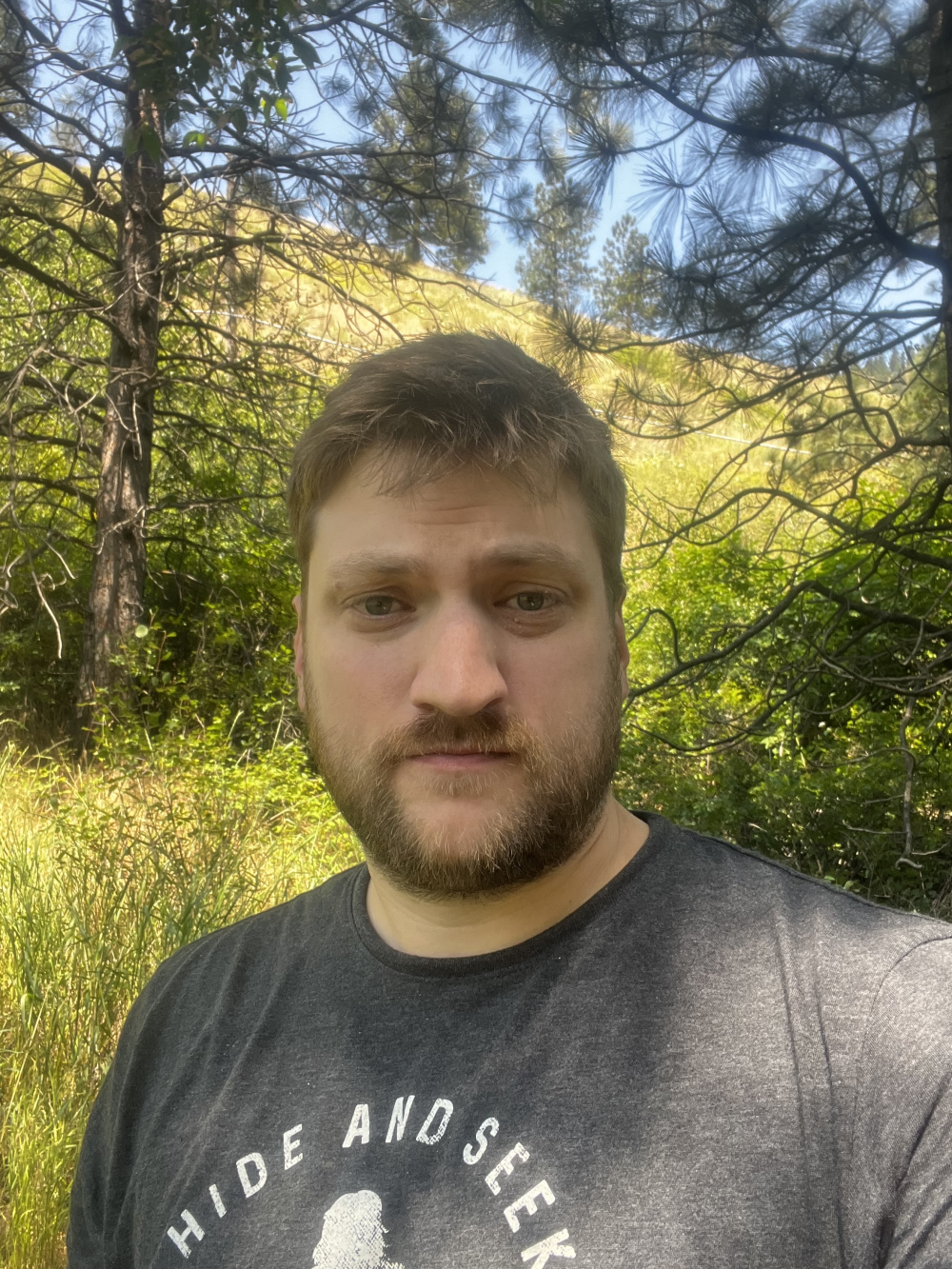}}]{David H. Rich}
    received his B.S. degree in Business and Mathematics from the University of Montana, Missoula, MT in 2019, and recieved his M.S. degree in Computer Science from the same institution in 2021.
    Since 2021, he has been working as a programmer in the Matsen Lab at Fred Hutchinson Cancer Center, Seattle, WA.
    His work has entailed contributing to projects involving phylogenetic inference.
\end{IEEEbiography}%

\begin{IEEEbiography}[{\includegraphics[width=1in,height=1.25in,clip]{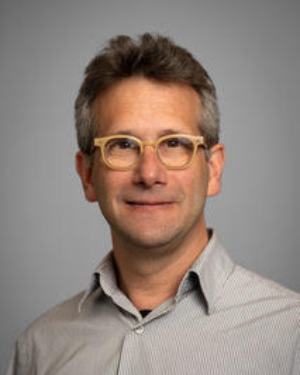}}]{Marc A. Suchard}
    received his B.A. in Biophysics from the University of California, Berkeley in 1995, followed by a Ph.D. in Biomathematics (2002) and M.D. (2004), both from UCLA. He is currently a Professor of Biostatistics and Human Genetics at the University of California, Los Angeles. Prof. Suchard's research focuses on scalable inference of stochastic processes for phylogenetics, pathogen evolution, and the analysis of large-scale electronic health records.
\end{IEEEbiography}

\begin{IEEEbiography}[{\includegraphics[width=1in,height=1.25in,keepaspectratio,clip]{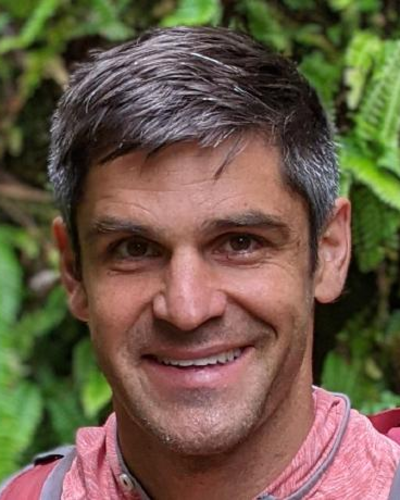}}]{Frederick A. Matsen IV}
received his B.S. degree in Mathematics from Stanford University in 2000, and his Ph.D. degree in Mathematics from Harvard University in 2006. 
He is currently a Professor at Fred Hutchinson Cancer Research Center and an HHMI Investigator. 
His research focuses on computational biology, phylogenetics, and immunology.
\end{IEEEbiography}

\clearpage
\section*{Supplementary Materials}
\beginsupplement

\section{Identifying PCM-Induced Subparsimonious Clades}
\label{sec:identifying_pcm_details}
Denote the most similar MP tree as $t^*$ and let $C_t, C_{t^*}$ be the set of clades in $t $ and $t^* $, respectively.
Then, we define the clades $C = C_t \cap C_{t^*}$ in both trees as \emph{jointly-MP with respect to $t^*$}, and the clades $C^{(t)} = C_{t} \setminus C$ in $t $ not present in $t^*$ as \emph{jointly-subparsimonious with respect to $t^*$}.
We pay special attention to the clades $C^{(t)}$ that have a corresponding clade in $C^{(t^*)}$ and systematically identify the differences between them.
Studying these differences sheds light on the primary mechanism by which simulated trees are subparsimonious.
To this end, our goals are to (i) find the largest set of clades in $t$ that are jointly-MP, and (ii) understand why the clades in $C^{(t)}$ are jointly-subparsimonious with respect to $t^*$.

The first objective can be formulated as solving the following optimization problem:
Let $T_{MP}$ be the set of MP trees on the dataset $X$.
\begin{equation}
  t^* = \argmin_{t' \in T_{MP}} d(t, t')
\end{equation}
where $d(t, t') = |C_t \triangle C_{t'} |$ is the RF distance between two trees.
We approach this problem by producing a large set of MP histories $\hat{T}_{MP}$ represented by a history sDAG $(V, E)$, as described in \autoref{sec:hdag_mp_tree_search}, trim it using the weight function defined in \autoref{sec:rf_distance_decomposes} to express only the minimum distance histories $\hat{T}^*_{MP}$.
If the minimizer is not unique, we trim the history sDAG further using \textbf{MinTrim} where our weight function is the summed sequence disimilarity between the simulated history and the histories in the history sDAG.
This is the sum of the hamming distances between the sequences of nodes that correspond to a clade that is present in both histories, and can be decomposed as sum of edge weight functions like RF distance.
Once we have our single most similar history $\hat{t}$, we address the second objective by mapping the clades of $C^{(t)}$ to similar clades in $C^{(\hat{t})}$, and analyzing the surrounding mutations to find the proportion of differing clades that are the result of small local deviations from an MP clade.

Specifically, we consider the collapsed versions of the trees, $t_c$ and $\hat{t}_c$, and find all the clades in $\hat{t}_c$ but not in $t_c$.
Then, for each differing clade $c \in C_{\hat{t}_c} \setminus C_{t_c}$, consider the node $v$ that corresponds to clade $c$ in $\hat{t}_c$.
Denote the parent of that node as $p$ and let $c_{p}$ be the clade corresponding to the parent.
If $c_p \in C_{t_c}$ then $t_c$ has a node that corresponds to $c_p$ as well.
Let's denote this node $p'$.
If the branches from $p'$ have multiple of the same mutation $m$, then $m$ is a PCM in $t$.
To determine whether the PCM corresponds to a structural difference between $t_c$ and $\hat{t}_c$, we check to see if grouping the branches that share $m$ under a single node in $t$ would recover the differing clade $c$.
If so, then that implies the simulated tree $t$ has at least one node that is the result of the mutations $m$ evolving independently, and that clade is a subset of $c_p$.
Thus, $c$ is the result of the PCM $m$.

\subsection{Minimizing Robinson-Foulds Distance in the History sDAG}
\label{sec:rf_distance_decomposes}

Before we can compare simulated trees to their most similar MP counterparts, we first must be able to find such a similar MP tree.
To do so, we employ an efficient trimming algorithm on the history sDAG, an innovation without which this paper would not be possible.
We start with a history sDAG containing as many MP trees on the simulated data as we can find.
\autoref{sec:hdag_mp_tree_search} details how we produce this initial history sDAG.

Next, we find the trees in this history sDAG which minimize Robinson-Foulds (RF) distance to the simulated tree.
Often, the history sDAG $(V, E)$ represents a set of histories $T$ so large that we could never hope to find the minimum distance tree by comparing them one-at-a-time.
Recall from~\autoref{sec:min_weight_trim} that \textbf{MinTrim}$(V, E, f)$ allows trimming with respect to a history weight function $f$ that decomposes as a sum over edges.
Here we show that RF distance to a given tree can be written in such a way, which implies that we can find the set of RF minimizing trees in $O(E)$ rather than $O(T)$.
We prove this in the following lemma.

\begin{lemma}
  Let $(V^*, E^*)$ be the history sDAG returned by $\textbf{MinTrim}\big(V, E, f\big)$ with $f(e) = 1 - 2\left (\1{C_t}(c_e)\right )$.
  Let $(V', E')$ be the history sDAG containing the histories in $(V, E)$ which minimize the RF distance to $t $.
  Then, $(V', E') = (V^*, E^*)$.
\end{lemma}

\begin{proof}
It suffices to show that there is a function $f$ such that $d(t, t') \propto \sum_{e \in t} f(e)$.
Informally, we want to show that we can represent a shifted version of RF distance with an edge decomposition.

Recall, that the RF distance between two rooted topologies $t $ and $t' $ is defined as $|C_t \triangle C_{t'} |$, where $C_t\subset \mathcal{P}(X) $ is the set of clades below nodes in the tree $t $.
Edges in a rooted tree are in bijection with clades via the assignment of each edge to the clade below its target node.
For an edge $e $ in the tree $t' $ let the clade associated to $e $ via this assignment be denoted $c_e $.
Therefore, a shifted version of rooted RF distance can be decomposed as a sum over edges of an arbitrary tree $t' $, with respect to a fixed reference tree $t $.
\begin{align*}
  d(t, t') &= |C_t \triangle C_{t'}|\\
           &= |C_{t'} \setminus C_t| - |C_{t'} \cap C_t| + |C_t|\\
           &= \sum_{c\in C_{t'}} \1{C_t^\sim}(c) - \sum_{c\in C_{t'}} \1{C_t}(c) + |C_t|\\
           &= |C_t| + \sum_{e\in t'} 1 - 2\left (\1{C_t}(c_e)\right )
\end{align*}
Since $|C_{t}|$ is constant with respect to the trees in the history sDAG, the tree(s) that minimize $d(t, t') - |C_{t}| $ also minimize $d(t, t')$.
\end{proof}

\subsection{Trimming}
\label{sec:min_trim_pseudocode}
Here we present pseudocode for \textbf{MinTrim} that was suggested by Definitions 14 and 15, and Lemma 11, from the paper that defined the history sDAG \cite{dumm2023representing}.
Let $(V, E)$ be a history sDAG that expresses the histories $T$, $f: E \rightarrow V$ be an edge weight, and $g:T \rightarrow \R$ such that for any history $(V', E') \in T $, $g(V', E') = \sum_{e \in E'} f(e)$.
The minimum weight of a subhistory below a given node $v = (\ell, U)$ is recursively defined as
\begin{align*}
  M_f(v) &= \sum_{C \in U} M_f(v, C)\\
  M_f(v, C) &= \min_{v_c \in \text{Ch}(v, C)}  M_f(v_c) + f(v, v_c)
\end{align*}
and $M_f(v) = 0$ if $v$ is a leaf node.

We precompute $M_f(v, C)$ and $M_f(v)$ with a single a post-order traversal, which requires visiting each edge once, and store them as maps.
Denote this operation, $\textbf{AnnotateMinWeight}(V, E)$. The full trimming algorithm is presented below.

\begin{algorithm}[H]
  \caption*{$\textbf{MinTrim}\big((V, E), f \big)$:
  Computes the minimum trim of a given a history sDAG $(V, E)$ with respect to edge weight function $f$.}
  \begin{algorithmic}
    \STATE $E^* = \{\}$, $V^* = \{\}$
    \STATE $M_f (v), M_f(v, C) \leftarrow \textbf{AnnotateMinWeight}(V, E)$

    \FOR {edge $e = (v, v_c)$ of $E$ in post-order}
    \IF{$M_f(v_c) + f(v,v_c) = M_f(v, C(v_c))$}
    \STATE add $e$ to $E^*$
    \ENDIF
    \ENDFOR

    \FOR{edge $e = (v, v_c)$ of $E$ in pre-order}
    \STATE add $v, v_c$ to $V^*$
    \ENDFOR

    \FOR {edge $e = (v, v_c)$ of $E$ in post-order}
    \IF {$v \not\in V^*$ or $v_c \not\in V^*$}
    \STATE remove $e$ from $E^*$
    \ENDIF
    \ENDFOR

    \STATE \textbf{return} $(V^*, E^*)$
  \end{algorithmic}
\end{algorithm}

This algorithm runs in $O(E)$, and by construction, \textbf{MinTrim} finds $(V^*, E^*)$ such that
\begin{align*}
  E' &= \left\{ (v,v_c) \in E \mid M_f(v_c) + f(v,v_c) = M_f(v, C(v_c))\right\},\\
  V^* &= \left\{ v \in V \mid \textit{v reachable from a path in } E' \right\}, and\\
  E^* &= \left\{(v,v_c) \in E' \mid v,v_c \in V^* \right\}.
\end{align*}
Lemma 11 of \cite{dumm2023representing} proves that $(V^*, E^*)$ does indeed contain the minimum weight histories.

\subsection{MinTrim Speedup}
We run \textbf{MinTrim} on history sDAGs (described in \autoref{sec:hdag_mp_tree_search}) built from our simulated data (described in \autoref{sec:simulations}) and compare the run time to naively searching for the minimum distance MP tree to the simulated one.
\textbf{Naive} finds the minimizing tree by iterating through each history in the history sDAG one at a time comparing the RF distance between that tree and the simulated tree.
As expected, this approach will take linear time in the number of histories contained in the history sDAG.
We observe that in practice \textbf{MinTrim} runs significantly faster, especially when the history sDAG contains many histories (\autoref{fig:scaling}).
For example, when there are $\approx 10^9$ histories, \textbf{MinTrim} computes the minimum in less than 10 minutes while \textbf{Naive} takes almost a day and a half.

Note that the runtime of \textbf{MinTrim} is a function of the number of edges in the history sDAG, not necessarily the number of histories.
This is why we see variation in the runtime for the same number of histories in Figure \ref{fig:scaling}.
In the worst case, when trees have no overlapping edges, \textbf{Naive} and \textbf{MinTrim} will have the same runtime.
However, typically many trees share many of the same edges.
\textbf{MinTrim} takes advantage of this to find the minimum distance maximum parsimony tree much more efficiently than would otherwise be possible without such a compact data structure as the history sDAG.

\begin{figure}[!t]
    \centering
    \includegraphics[width=0.5\textwidth]
            {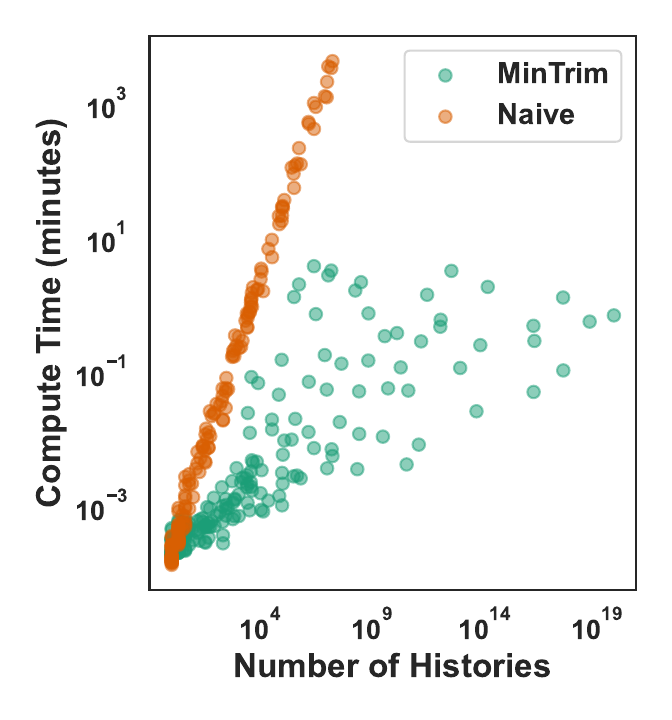}
    \caption{\
        Time in minutes to run \textbf{MinTrim} versus the naive approach on the history sDAG, which finds the minimum distance tree by comparing histories to the reference tree one at a time.
        The naive approach was cut off after two days.
    }%
    \label{fig:scaling}
\end{figure}

\subsection{Branch-Length Sensitivity}
\label{sec:branch-len-sensitivity}
We stress that our main result that PCMs compose the primary structural deviation between the true tree and an MP tree \emph{only applies in the densely sampled regime}.
To examine the importance of the small-branch assumption, we repeat the PCM-proportion experiment described in Section \ref{sec:identifying_pcms} for the AY.34.2 clade by performing 5 independent simulations on the same trees, but with branch lengths multiplied by 1, 2, 4, and 8 (Fig \ref{fig:pcm_vs_branch}).
For the original branch lengths, the median proportion of differing nodes that can be explained by PCMs is 81.8\%.
As we increase the branch lengths by a factor of 2, then 4, then 8, the median proportion decreases all the way to 52.6\%.
Based on these results, it appears that for less densely sampled data (i.e., trees with longer branches), the true tree's deviation from the set of MP trees is less explainable by PCMs.

\begin{figure}[!t]
    \centering
    \includegraphics[width=0.5\textwidth]
            {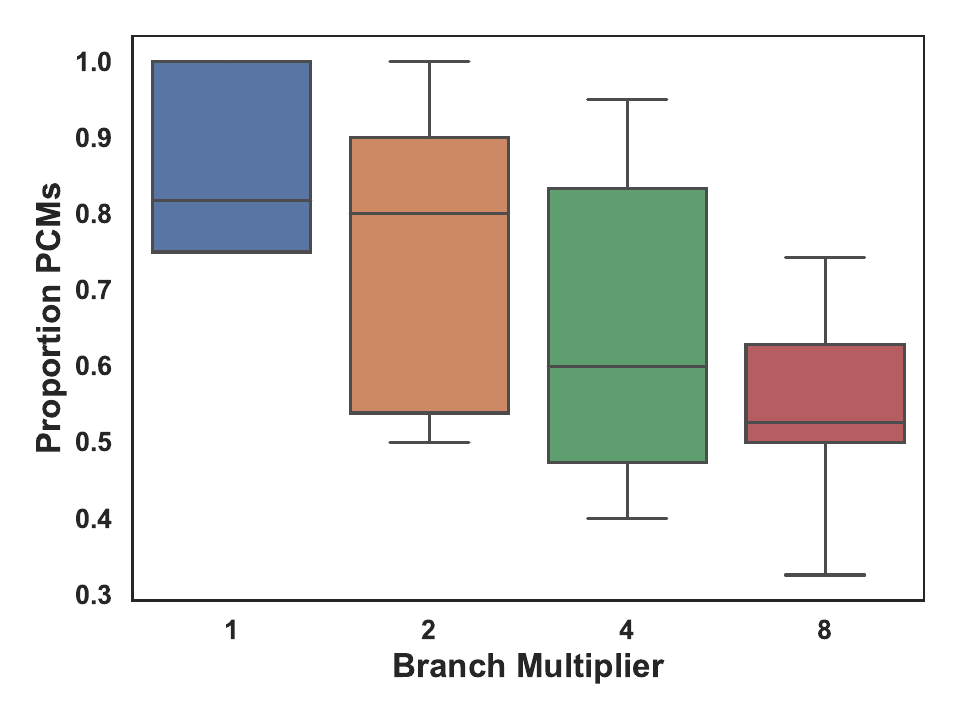}
    \caption{\
        Proportion of differences due to PCMs for various branch length multipliers. Each boxplot is taken over 5 independent trials.
    }%
    \label{fig:pcm_vs_branch}
\end{figure}

\clearpage
\section{Tuning Simulation Hyperparameters}

\begin{figure*}[!t]
  \centering
  \includegraphics[width=0.75\textwidth]{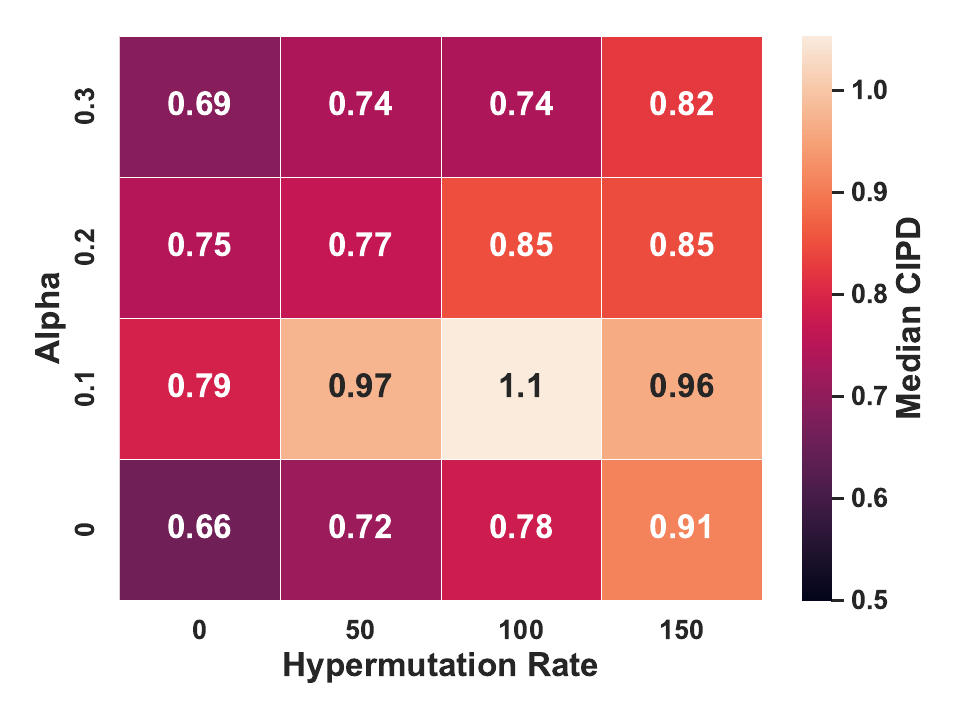}
  \caption{\
  Heatmap of median CIPD for each combination of $\alpha$ and $r$ across all simulated datasets.
  When $\alpha = 0$ or $r = 0$, that indicates that we do not include the respective form of rate variation.
  When the CIPD is 1, the PD of simulated data matches that of the corresponding real data.
  }%
  \label{fig:pd_heatmap}
\end{figure*}

Simulation parameters $(\alpha=0.1, r=50)$ yield among the best median CIPD values across all simulations (\autoref{fig:pd_heatmap}).
In addition to having a good median CIPD, we found that the variance in CIPD among simulations with $\alpha=0.1, r=50$ was lower than that of the other good parameter choices $(\alpha=0.1, r=100)$ and $(\alpha=0.1, r=150)$.

\section{CIPD for MAPLE Simulated Data}
\label{sec:cipd_maple}
In \cite{MAPLE2022}, the authors benchmark their likelihood-based phylogenetic inference software by extracting subtrees from the \usher{} global tree and simulating mutations on that topology using phastSim \cite{2022phastsim} under three classes of simulation:
\begin{itemize}
\item The ``basic'' simulation scenario uses the GTR model with no rate variation and full genomes available.

\item The ``rate variation'' scenario uses the GTR+G model with four genome site categories, all with the same frequency and with relative substitution rates of 0.1, 0.5, 1 and 2.

\item The ``sequence ambiguity'' scenario modifies the basic scenario to include ambiguous characters.
\end{itemize}

We analyze the CIPD under similar settings for the second type of simulation.
Specifically, we repeat the simulations described in \autoref{sec:simulations} including topology selection and simulation with phastSim, but use categorical Gamma distributed rate variation and no hypermutation.
Additionally, we use the UNREST substitution model instead of GTR.
CIPD is computed as described in \autoref{sec:simulations}.

\begin{figure*}[!t]
    \centering
    \includegraphics[width=0.75\textwidth]{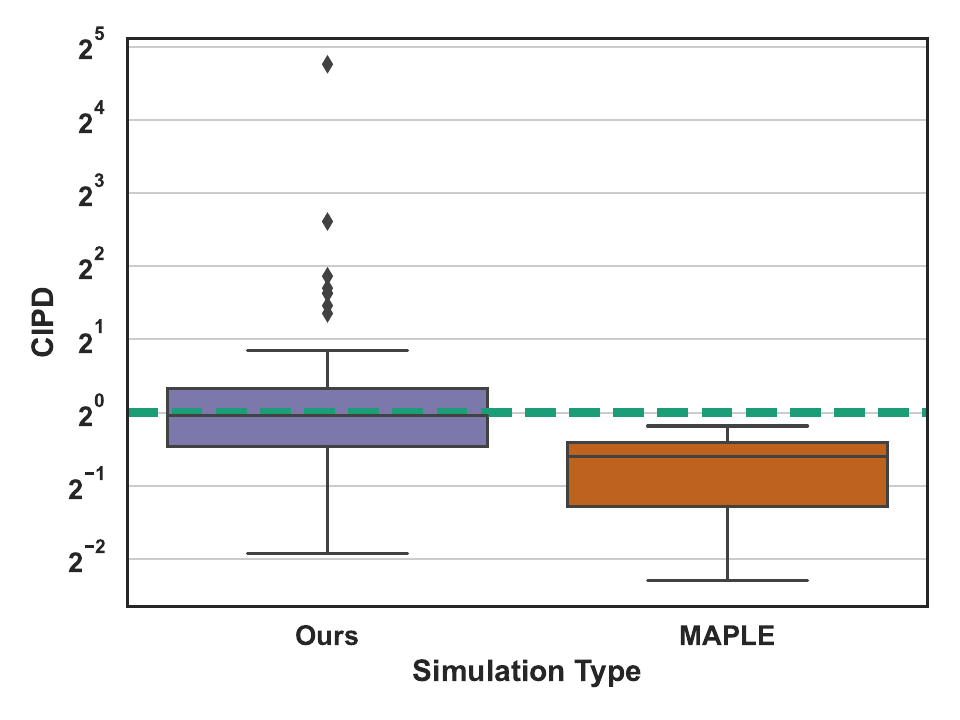}
    \caption{\
    The boxplots show CIPD values across all 200 simulations.
    The left boxplot uses our simulation settings with hypermutation at a rate of $r=50$ and Gamma distributed rates using $\alpha=0.1$.
    The right boxplot simulates without hypermutation and relative substitution rates of 0.1, 0.5, 1 and 2.
    The green dashed line indicates the CIPD value of simulated data with realistic PD.
    }%
    \label{fig:MAPLE_CIPD}
  \end{figure*}

We find that the CIPD values of the MAPLE datasets are significantly below 1 (\autoref{fig:MAPLE_CIPD}), indicating unrealistic levels phylogenetic inferential difficulty.
All the simulations had a CIPD below 90\% and half of them had a CIPD value below 70\%.
While some of our simulations have significantly higher CIPD than 1 (\autoref{fig:MAPLE_CIPD}), our distribution of CIPD is centered at 1 and almost all of our simulations have a CIPD value that differs from the ideal value, by no more than a factor of 2.
It appears that hypermutation is important for achieving realistic levels of PD.
We see further evidence of this in \autoref{fig:pd_heatmap} where all types of simulations without hypermutation (left-most column of heatmap) have a median CIPD less than 0.8.
While we did not investigate the PD realism of the ``basic'' and ``sequence ambiguity'' simulations from \cite{MAPLE2022}, we expect these types of simulations to have CIPD values that are at least as small as the ``rate variation'' scenario because they do not include rate variation.
We do however note that in \cite{MAPLE2022}, the authors simulate sequences on \emph{much larger} \usher{}-clades than the ones that we use (i.e., between 2,000-20,000 tips).
So, our analysis might not generalize to trees of that size.

Creating accurate and fair simulations is challenging, and the best method depends on the goal.
In our work, we focus on simulating as realistic of data as possible because that is essential to generalizing our conclusions to real data sets.
However, in \cite{MAPLE2022}, the authors use their simulations to compare their model's performance to other likelihood-based methods.
They chose to not use more realistic simulations out of fairness since the other methods don't include UNREST and typically use GTR with four categories for rate variation.

\section{MrBayes Parsimony Posterior}

\begin{figure*}[!t]
    \centering
    \includegraphics[width=0.47\textwidth]{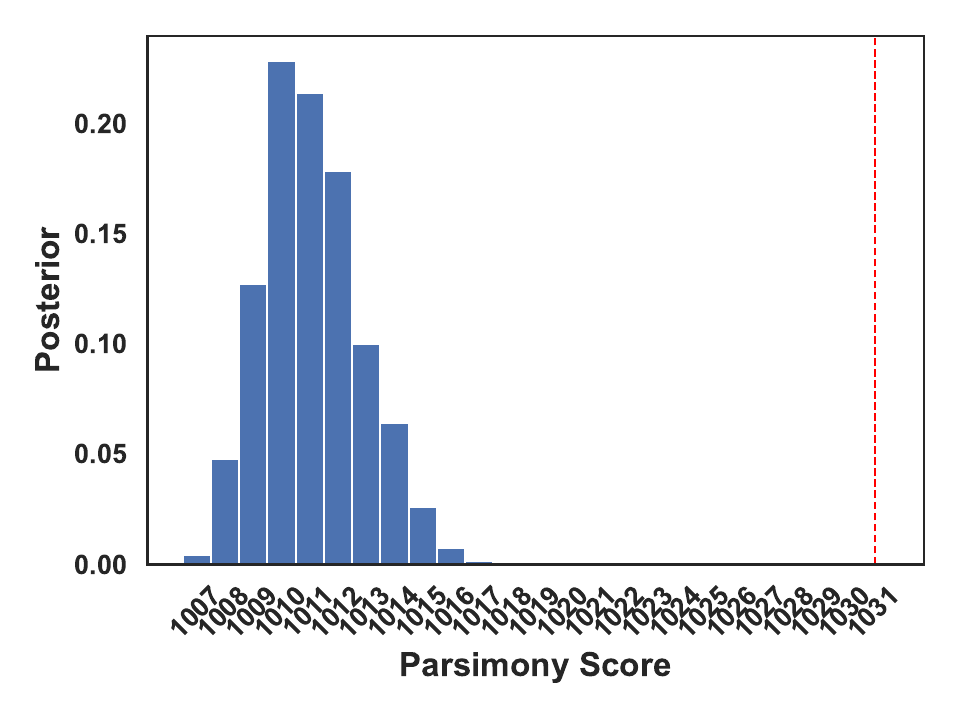}
    \includegraphics[width=0.47\textwidth]{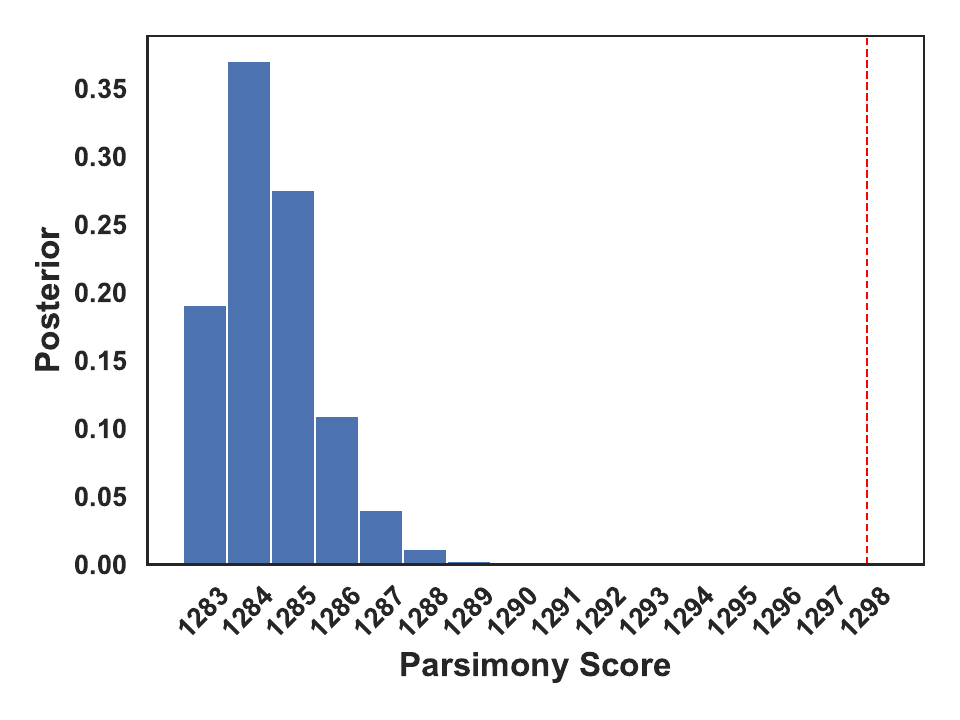}
    \caption{\
       Histogram showing the distribution of parsimony scores of topologies sampled with MrBayes, compared to the best possible parsimony score of the simulated topology, shown by the red dotted line.
       The topologies used for these simulations correspond to the subtree of \usher{}-clades AY.108 and P.1.7.
    }%
    \label{fig:parsimony_posterior}
\end{figure*}

We simulate data as described in \autoref{sec:simulations} and run MrBayes on the resulting set of sequences with the most general model settings available.
This includes using the GTR substitution model and Gamma distributed rate variation.

In \autoref{fig:parsimony_posterior}, the red line shows the best parsimony score of the tree that the sequences were simulated on, and the histogram shows the distribution of parsimony scores in trees sampled from the MrBayes posterior.
The MrBayes samples are significantly more parsimonious than the simulated tree.
We suspect that this occurs because our inference model assumptions don't match the model under which we're simulating data.
One example of this is that we include hypermutation in our simulations, whereas MrBayes doesn't model hypermutation.
However, hypermutation appears to be a realistic feature of the data, so given that MrBayes infers a posterior that puts too much weight on high parsimony trees in these simulations, we expect these issues to carry over to real data as well.

\end{document}